\documentclass{LMCS}

\def\dOi{9(4:5)2013}
\lmcsheading%
{\dOi}
{1--36}
{}
{}
{Feb.~\phantom05, 2013}
{Oct.~16, 2013}
{}

\usepackage{afterpage}  %
\usepackage{color}
\usepackage{pgf}
\usepackage{pgfarrows}
\usepackage{tikz}\usetikzlibrary{arrows,shapes,calc,automata,positioning,external}
\usepackage{amsmath,amssymb,amsthm}
\usepackage{latexsym}
\usepackage{amsmath}
\usepackage{amsthm}
\usepackage{amssymb}
\usepackage{xspace}
\usepackage{bussproofs}
\usepackage{microtype}
\usepackage{relsize}
\usepackage{booktabs}
\usepackage{multirow}
\usepackage{enumitem}
\usepackage{graphics}
\usepackage{tabularx}
\usepackage{diagbox} %
\usepackage{ifpdf}
\ifpdf\usepackage{hyperref}\else\usepackage[dvips]{hyperref}\fi

\definecolor{darkblue}{rgb}{0,0,0.6}
\hypersetup{
	breaklinks=true,
}
\ifpdf\hypersetup{
}\else\hypersetup{
	citecolor=black,
	linkcolor=black,
	urlcolor=black
}\fi

\tikzstyle{every picture}+=[remember picture]
\newcommand{\tikzlabel}[2]{\tikz[baseline, inner sep=0pt, outer sep=0pt, anchor=base]{\node[inner sep=0pt, outer sep=0pt](#1){#2};}}
\newcommand{\tikzoverlay}[1]{
\begin{tikzpicture}[overlay]
#1
\end{tikzpicture}
}
\newcommand{\tikzedge}[2]{
	\path[->] (#1) edge [draw=black] (#2);
}
\newcommand{\tikzloopedge}[2]{
	\path[->] (#1) edge [draw=black,in=0,out=0] (#2);
}
\newcommand{\tikzloopedgeleft}[2]{
	\path[->] (#1) edge [draw=black,in=180,out=180] (#2);
}

\def\myarm{1cm}
\def\myangle{0}
\tikzset{
  arm/.default=1cm,
  arm/.code={\def\myarm{#1}}, %
  angle/.default=0,
  angle/.code={\def\myangle{#1}} %
}

\tikzset{
    myncbar/.style = {to path={
        let
            \p1=($(\tikztotarget)+(\myangle:\myarm)$)
        in
            -- ++(\myangle:\myarm) coordinate (tmp)
            -- ($(\tikztotarget)!(tmp)!(\p1)$)\tikztonodes
            -- (\tikztotarget)%
    }}
}

\newtheorem{theorem}{Theorem}
\newtheorem{proposition}[theorem]{Proposition}
\newtheorem{lemma}[theorem]{Lemma}
\newtheorem{corollary}[theorem]{Corollary}

\theoremstyle{definition}
\newtheorem{definition}[theorem]{Definition}
\newtheorem{example}[theorem]{Example}

\newcommand{\refsection}[1]{Section~\ref{#1}}
\newcommand{\refsubsection}[1]{Subsection~\ref{#1}}

\newcommand{\EXPTIME}{\textsmaller{EXPTIME}\xspace}

\newcommand{\nxt}{\ensuremath{\mathtt{X}}\xspace}
\newcommand{\until}{\ensuremath{\mathtt{U}}\xspace}
\newcommand{\release}{\ensuremath{\mathtt{R}}\xspace}
\newcommand{\allpath}{\ensuremath{\mathtt{A}}\xspace}
\newcommand{\qpath}{\ensuremath{Q}\xspace}
\newcommand{\expath}{\ensuremath{\mathtt{E}}\xspace}
\newcommand{\finally}{\ensuremath{\mathtt{F}}\xspace}
\newcommand{\generally}{\ensuremath{\mathtt{G}}\xspace}

\newcommand{\true}{\ensuremath{\mathtt{t\!t}}}
\newcommand{\false}{\ensuremath{\mathtt{f\!f}}}

\newcommand{\fl}[1]{\ensuremath{\mathit{FL}(#1)}}
\newcommand{\flr}[1]{\ensuremath{\mathit{FL}_{\release}(#1)}}
\newcommand{\subf}[1]{\ensuremath{\mathit{Sub}(#1)}}
\newcommand{\goals}[1]{\ensuremath{\mathit{Conf}(#1)}}

\newcommand{\Transsys}{\ensuremath{\mathcal{T}}}
\newcommand{\States}{\ensuremath{\mathcal{S}}}
\newcommand{\Prop}{\ensuremath{\mathcal{P}}}

\newcommand{\ctlstar}{\textup{CTL}\ensuremath{{}^{*}}\xspace}
\newcommand{\ctlplus}{\textup{CTL}\ensuremath{{}^{+}}\xspace}
\newcommand{\ctl}{\textup{CTL}\xspace}
\newcommand{\ltl}{\textup{LTL}\xspace}
\newcommand{\pdl}{\textup{PDL}\xspace}

\newcommand{\Andrule}{\ensuremath{\mathtt{(}\!\land\!\mathtt{)}}}
\newcommand{\Orrule}{\ensuremath{\mathtt{(}\!\lor\!\mathtt{)}}}

\newcommand{\AAndrule}{\ensuremath{\mathtt{(A}\!\land\!\mathtt{)}}}
\newcommand{\EAndrule}{\ensuremath{\mathtt{(E}\!\land\!\mathtt{)}}}
\newcommand{\AOrrule}{\ensuremath{\mathtt{(A}\!\lor\!\mathtt{)}}}
\newcommand{\EOrrule}{\ensuremath{\mathtt{(E}\!\lor\!\mathtt{)}}}
\newcommand{\ALitrule}{\ensuremath{\mathtt{(Al)}}}
\newcommand{\ELitrule}{\ensuremath{\mathtt{(El)}}}
\newcommand{\AArule}{\ensuremath{\mathtt{(AA)}}}
\newcommand{\EArule}{\ensuremath{\mathtt{(EA)}}}
\newcommand{\AErule}{\ensuremath{\mathtt{(AE)}}}
\newcommand{\EErule}{\ensuremath{\mathtt{(EE)}}}
\newcommand{\AUrule}{\ensuremath{\mathtt{(AU)}}}
\newcommand{\AFrule}{\ensuremath{\mathtt{(AF)}}}
\newcommand{\EUrule}{\ensuremath{\mathtt{(EU)}}}
\newcommand{\EFrule}{\ensuremath{\mathtt{(EF)}}}
\newcommand{\ARrule}{\ensuremath{\mathtt{(AR)}}}
\newcommand{\AGrule}{\ensuremath{\mathtt{(AG)}}}

\newcommand{\ERrule}{\ensuremath{\mathtt{(ER)}}}
\newcommand{\EGrule}{\ensuremath{\mathtt{(EG)}}}
\newcommand{\Ettrule}{\ensuremath{\mathtt{(Et\!t)}}}

\newcommand{\XruleNoE}{\ensuremath{\mathtt{(X_0)}}}
\newcommand{\XruleWithE}{\ensuremath{\mathtt{(X_1)}}}

\newcommand{\placeholder}{\underline{\phantom X}}

\newcommand{\Nat}{\ensuremath{\mathbb{N}}}

\newcommand{\sideFormulas}{\Phi}   %
\newcommand{\addSideFormulas}[1]%
  {#1\def\addSideFormulasTest{#1}\ifx\addSideFormulasTest\empty\sideFormulas\else, \sideFormulas\fi}

\newcommand{\unaryRule}[3]{%
  \AxiomC{\ensuremath{\addSideFormulas{#2}}}%
  \LeftLabel{#1}%
  \UnaryInfC{\ensuremath{\addSideFormulas{#3}}}%
  \bottomAlignProof\DisplayProof}

\newcommand{\choiceRule}[4]{%
  \AxiomC{\ensuremath{\addSideFormulas{#2}}}%
  \AxiomC{\ensuremath{\hspace*{-5mm}|\hspace*{-5mm}}}%
  \AxiomC{\ensuremath{\addSideFormulas{#3}}}%
  \LeftLabel{#1}%
  \TrinaryInfC{\ensuremath{\addSideFormulas{#4}}}%
  \bottomAlignProof\DisplayProof}
\newcommand{\branchingRule}[4]{%
  \AxiomC{\ensuremath{\addSideFormulas{#2}}}%
  \AxiomC{\ensuremath{\hspace*{-6mm}|\hspace*{2mm}\cdots\hspace*{2mm}|\hspace*{-6mm}}}%
  \AxiomC{\ensuremath{\addSideFormulas{#3}}}%
  \LeftLabel{#1}%
  \TrinaryInfC{\ensuremath{\addSideFormulas{#4}}}%
  \bottomAlignProof\DisplayProof}

\ifpdf
\else
\fi

\begin{document}

\title[Satisfiability Games for Branching-Time Logics]{Satisfiability Games for Branching-Time Logics\rsuper*}

\author[O.~Friedmann]{Oliver Friedmann\rsuper a}
\address{{\lsuper{a,b}}Department of Computer Science, Ludwig-Maximilians-University Munich, Germany}
\email{\{oliver.friedmann, markus.latte\}@ifi.lmu.de}

\author[M.~Latte]{Markus Latte\rsuper b}
\address{\vspace{-18 pt}}

\author[M.~Lange]{Martin Lange\rsuper c}
\address{{\lsuper c}School of Electrical Engineering and Computer Science, University of Kassel, Germany}
\email{martin.lange@uni-kassel.de}
\thanks{Financial support was provided by the DFG Graduiertenkolleg 1480 (PUMA) and the European Research Council 
under the European Community's Seventh Framework Programme (FP7/2007-2013) / ERC grant agreement no 259267.}

\ACMCCS{[{\bf Theory of computiation}]:
Logic---Modal and temporal logics; Computational complexity and
cryptography---Complexity theory and logic }
\titlecomment{{\lsuper*}A preliminary version appeared as \cite{ijcar2010}.} 

\begin{abstract}
The satisfiability problem for branching-time temporal logics like
\ctlstar, \ctl and \ctlplus has important applications in program
specification and verification. Their computational complexities are
known: \ctlstar and \ctlplus are complete for doubly exponential time,
\ctl is complete for single exponential time. Some decision procedures
for these logics are known; they use tree automata, tableaux or axiom
systems.

In this paper we present a uniform game-theoretic framework for the
satisfiability problem of these branching-time temporal logics. We
define satisfiability games for the full branching-time temporal logic
\ctlstar using a high-level definition of winning condition that
captures the essence of well-foundedness of least fixpoint
unfoldings. These winning conditions form formal languages of
$\omega$-words. We analyse which kinds of deterministic
$\omega$-automata are needed in which case in order to recognise these
languages. We then obtain a reduction to the problem of solving parity
or B\"uchi games.  The worst-case complexity of the obtained
algorithms matches the known lower bounds for these logics.

This approach provides a uniform, yet complexity-theoretically optimal
treatment of satisfiability for branching-time temporal logics. It
separates the use of temporal logic machinery from the use of automata
thus preserving a syntactical relationship between the input formula
and the object that represents satisfiability, i.e.\ a winning
strategy in a parity or B\"uchi game. The games presented here work on
a Fischer-Ladner closure of the input formula only. Last but not
least, the games presented here come with an attempt at providing tool
support for the satisfiability problem of complex branching-time
logics like \ctlstar and \ctlplus.
\end{abstract}
\maketitle

\section{Introduction}
\label{sec:intro}

The full branching-time temporal logic \ctlstar is an important tool for the specification and verification 
of reactive~\cite{journals/igpl/GabbayP08} or agent-based systems~\cite{conf/atal/LuoSSCL05}, and  
for program synthesis~\cite{Pnueli88}, etc. Emerson and Halpern have introduced 
\ctlstar~\cite{Emerson:1986:SNN} as a formalism which supersedes both the branching-time logic \ctl~\cite{ClEm81b} 
and the linear-time logic \ltl~\cite{Pnueli:1977}.  

\subsubsection*{Automata-theoretic approaches.}
As much as the introduction of \ctlstar has led to an easy unification of \ctl and \ltl, it has also proved to be quite a difficulty in 
obtaining decision procedures for this logic. The first procedure by Emerson and Sistla was 
automata-theoretic~\cite{IC::EmersonS1984} and roughly works as follows. A formula is translated into a doubly-exponentially
large tree automaton whose states are Hintikka-like sets of sets of subformulas of the input formula. This
tree automaton recognises a superset of the set of tree models of the input formula. It is lacking a mechanism
that ensures that certain temporal operators are really interpreted as least fixpoints of certain monotone
functions rather than arbitrary fixpoints. Such a mechanism is provided by intersecting this automaton with
a tree automaton that recognises a language which is defined as the set of all trees such that all paths in
such a tree belong to an $\omega$-regular word language, recognised by some B\"uchi automaton for instance.
In order to turn this into a tree automaton, it has to be determinised first. A series of improvements
on B\"uchi automata determinisation for this particular word language has eventually led to Emerson and Jutla's 
automata-theoretic decision procedure~\cite{Emerson:2000:CTA} whose asymptotic worst-case running time is optimal, 
namely doubly exponential~\cite{STOC85*240}.  
This approach has a major drawback though, as noted by Emerson~\cite{Emerson90a}: 
``{\em \ldots due to the delicate combinatorial constructions involved, there is usually no clear relationship
between the structure of automaton and the candidate formula.}'' 

The constructions he refers to are determinisation and complementation of B\"uchi automata. Determinisation in
particular is generally perceived as the bottleneck in applications that need deterministic automata for
$\omega$-regular languages. A lot of effort has been spent on attempts to avoid B\"uchi determinisation for
temporal branching-time logics. Kupferman, Vardi and Wolper introduced alternating automata \cite{Muller:1987:AAI} 
for branching-time temporal logics \cite{BVW94,KVW00}. The main focus of this approach was the model-checking problem
for such logics, though. While satisfiability checking and model checking for linear-time temporal logics are
virtually the same problem and therefore can be handled by the same machinery, i.e.\ class of automata and
algorithms, the situation for branching-time temporal logics is different. In the automata-theoretic framework,
satisfiability corresponds to the general emptiness problem whereas model-checking reduces to the simpler
1-letter emptiness problem. Still, alternating automata provide an alternative framework for the satisfiability
checking problem for branching-time logics, and some effort has been paid in order to achieve emptiness checks,
and therefore satisfiability checking procedures. Most notably, Kupferman and Vardi have suggested a way to test
tree automata for emptiness which avoids B\"uchi determinisation \cite{conf/focs/KupfermanV05}. However, it is based
on a satisfiability-preserving reduction only rather than an equivalence-preserving one. Thus, it avoids the 
``{\em delicate combinatorial constructions}'' which are responsible for the lack of a ``{\em clear relationship 
between the structure of automaton and the candidate formula}'' in Emerson's view~\cite{Emerson90a}, but to 
avoid these constructions it gives up any will to preserve such a clear relationship.

\subsubsection*{Other approaches.}
Apart from these automata-theoretic approaches, a few different ones have been presented as well. For instance, there
is Reynolds' proof system for validity \cite{Reynolds00}. Its completeness proof is rather intricate and relies on 
the presence of a rule which violates the subformula property. In essence, this rule quantifies over an arbitrary
set of atomic propositions. Thus, while it is possible to check a given tree for whether ot not it is a proof for 
a given \ctlstar formula, it is not clear how this system could be used in order to find proofs for given \ctlstar
formulas.

Reynolds has also presented a tableaux system for \ctlstar \cite{conf/fm/Reynolds09,Rey:startab} which shares some 
commonalities with the automata-theoretic approach by Emerson and others as well as the game-based approach presented 
here. However, one of the main differences between tableaux on one side and automata and games on the other has a major
effect in the case of such a complex branching-time logic: while automata- and game-based approaches typically separate
the characterisation (e.g.\ tree automaton or parity game) from the algorithm (e.g.\ emptiness test or check for 
winning strategy), tableaux are often designed monolithically, i.e.\ with the characterisation and algorithm as one
procedure. As a result, Reynolds' tableaux rely on some repetition test which, done in a na\"{\i}ve way, is hopelessly
inefficient in practice. On the other hand, it is not immediately clear how a more clever and thus more efficient 
repetition check could be designed for these tableaux, and we predict that it would result in the introduction of
B\"uchi determinisation. 

A method that is traditionally used for predicate logics is resolution. It has also been used to devise decision 
procedures for temporal logics, starting with the linear-time temporal logic \ltl \cite{ijcai91*99}, followed by 
the simple branching-time temporal logic \ctl \cite{journals/jetai/BolotovF99,journals/aicom/ZhangHD10}. Finally, 
there is also a resolution-based approach to \ctlstar which combines linear-time temporal logic resolution with 
additional techniques to handle path quantification \cite{conf/mfcs/BolotovDF99}. However, all resolution methods 
rely on the fact that the input formula is transformed into a specialised normal form. The known transformations
are not trivial, and they only produce equi-satisfiable formulas. Thus, such methods also do not preserve a close
connection between the models of the input formula and its subformulas.

\subsubsection*{The game-based framework.}
In this paper we present a game-based characterisation of \ctlstar satisfiability. In such games, two players
play against each other with competing objectives: player 0 should show that the input formula is satisfiable 
whereas player 1 should show that it is not. 
Formally, the \ctlstar satisfiability game for some input formula
is a graph of doubly exponential size on which the two players move a token along its edges. There is a 
winning condition in the form of a formal language of infinite plays which describes the plays that are won by
player 0. This formal language turns out to be $\omega$-regular, and it is known that arbitrary games with such
a winning condition can be solved by a reduction to parity games. 
This yields an asymptotically optimal decision procedure. Still, the games only use subformulas of the input formula,
and automata are only needed in the actual decision procedure but not in the definition of the satisfiability games as
such. Thus, it moves the ``{\em delicate combinatorial constructions}'' to a place where they do not destroy a 
``{\em clear relationship between the [\ldots] input formula}'' and the parity game anymore. This is very useful 
in the setting of a user interacting with a satisfiability checker or theorem prover for \ctlstar, when they may want 
to be given a reason for why a formula is not satisfiable for instance. 

The delicate combinatorial procedures, i.e.\ B\"uchi determinisation and complementation is kept at minimum by
analysing carefully where it is needed. We decompose the winning condition such that the transformation of a nondeterministic
B\"uchi into a deterministic parity automaton \cite{FOCS88*319,conf/lics/Piterman06,conf/icalp/KahlerW08,conf/fossacs/Schewe09} 
is only needed for some part. The other is handled directly using manually defined deterministic automata.

We also consider two important fragments of \ctlstar, namely the well-known \ctl and the lesser known \ctlplus. 
The former has less expressive power and is computationally simpler: \ctl satisfiability is complete for 
deterministic singly exponential time only~\cite{EmersonHalpern85}. The latter already carries the full complexity 
of \ctlstar despite sharing its expressive power with the weaker \ctl~\cite{EmersonHalpern85}: \ctlplus satisfiability 
is also complete for doubly exponential time~\cite{JL-icalp03}. The simplicity of \ctl when compared to \ctlstar 
also shows through in this game-based approach. The rules can be simplified a lot when only applied to \ctl formulas,
resulting in an exponential time procedure only. Even more so, the simplification gets rid of the need for 
automata determinisation procedures at all. Again, it is possible to construct a very small and deterministic 
B\"uchi automaton directly that can be used to check the winning conditions when simplified to \ctl formulas. 

The computational complexity of \ctlplus suggests that no major simplifications in comparison to \ctlstar are
possible. Still, an analysis of the combinatorics imposed by \ctlplus formulas on the games shows that for such
formulas it suffices to use determinisation for co-B\"uchi automata~\cite{Miyano:1984:AFA} instead of that for B\"uchi 
automata. This yields asymptotically smaller automata, is much easier to implement and also results in B\"uchi games 
rather than general parity games.

\subsubsection*{Advantages of the game-based approach.}
The game-theoretic framework achieves the following advantages.
\begin{itemize}[leftmargin=1em]
\item[--] The framework uniformally treats the standard branching-time logics from the relatively simple \ctl to the
      relatively complex \ctlstar.
\item[--] It yields \emph{complexity-theoretic optimal} results, i.e.\ satisfiability checking using this framework
      is possible in exponential time for \ctl and doubly exponential time for \ctlstar and \ctlplus.
\item[--] Like the automata-theoretic approaches, it separates the characterisation of satisfiability through a syntactic
      object (a parity game) from the test for satisfiability (the problem of solving the game). Thus, advances in the 
      area of parity game solving carry over to satisfiability checking.
\item[--] Like the tableaux-based approach, it keeps a very close relationship between the input formula and the structure
      of the parity game thus enabling feedback from a (counter-)model for applications in specification and verification.
\item[--] Satisfiability checking procedures based on this framework are implemented in the\break \textsc{MLSolver} platform%
      ~\cite{FriedmannL:MLSolver} which uses the high-performance parity game solver \textsc{PGSolver}~\cite{fl-atva09} as 
      its algorithmic backbone --- see the corresponding remark about the separation between characterisation and algorithm 
      above.
\end{itemize}

\subsubsection*{Organisation.} The rest of the paper is organised as follows. \refsection{sec:ctlstar} recalls \ctlstar.
\refsection{sec:tableaux} presents the satisfiability games. \refsection{sec:correctness} gives the formal
soundness and completeness proofs for the presented system. \refsection{sec:decproc} describes the
decision procedure, i.e.\ the reduction to parity games. \refsection{sec:fragments} presents the simplifications
one can employ in both the games and the reduction when dealing with formulas of \ctl, respectively \ctlplus. 
\refsection{sec:compare} compares the games presented here with other decision procedures for branching-time logics,
in particular with respect to technical similarities, pragmatic aspects, results that follow from them, etc. 
\refsection{sec:further} concludes with some remarks on possible further work into this direction.

\section{The Full Branching Time Logic}
\label{sec:ctlstar}

Let $\Prop$ be a countably infinite set of propositional constants. A transition system is a tuple
$\Transsys = (\States,s^*,\to,\lambda)$ with $(\States,\to)$ being a directed graph, $s^* \in \States$ being
a designated starting state and
$\lambda: \States \to 2^{\Prop}$ is a labeling function. We assume transition systems to be total, i.e.\
every state has at least one successor.
A \emph{path} $\pi$ in $\Transsys$ is an infinite sequence of states $s_0,s_1,\ldots$ s.t.\ $s_i \to s_{i+1}$
for all $i$. With $\pi^k$ we denote the suffix of $\pi$ starting with state $s_k$, and $\pi(k)$ denotes $s_k$
in this case. 

Branching-time temporal formulas in negation normal form\footnote{Alternatively, we could have admitted negations 
everywhere---not only in front of a proposition. However, for any formula of one form there is an equivalent 
and linearly sized formula of the other form: just apply De Morgan's laws to the binary propositional connectors, 
e.g.\ $\neg (\varphi_1 \wedge \varphi_2) \equiv (\neg \varphi_1) \vee (\neg \varphi_2)$, fixpoint duality to fixpoints, 
e.g.\ $\neg(\varphi_1 \until \varphi_2) \equiv (\neg \varphi_1) \release (\neg \varphi_2)$ and the property 
$\neg \nxt \varphi \equiv \nxt \neg \varphi$.
} are given by the following grammar.
\begin{displaymath}
  \varphi \quad ::= \quad \true \mid \false \mid p \mid \neg p \mid \varphi \lor \varphi \mid \varphi \land \varphi \mid  \nxt\varphi \mid
    \varphi \until \varphi  \mid \varphi \release \varphi \mid \expath\varphi \mid \allpath\varphi
\end{displaymath}
where $p \in \Prop$. Formulas of the form $\true$, $\false$, $p$ or $\neg p$ are called \emph{literals}. 

Boolean constructs other than conjunction and disjunction, like $\to$ for instance, are derived as usual. Temporal
operators other than the ones given here are also defined as usual: $\finally \varphi := \true \until \varphi$ and
$\generally \varphi := \false \release \varphi$.

The set of \emph{subformulas} of a formula $\varphi$, written as $\subf{\varphi}$, is defined as usual, in particular the set contains $\varphi$.
In contrast, a formula $\psi$ is a \emph{proper subformula} of $\varphi$ if both
are different and $\psi$ is a subformula of $\varphi$.
The \emph{Fischer-Ladner} closure of $\varphi$ is the least set $\fl{\varphi}$ that %
is closed under taking subformulas, and contains, for each $\psi_1 \until \psi_2$ or $\psi_1 \release \psi_2$, also the formulas
$\nxt(\psi_1 \until \psi_2)$ respectively $\nxt(\psi_1 \release \psi_2)$. Note that $|\fl{\varphi}|$ is at most twice
the number of subformulas of~$\varphi$. Let $\flr{\varphi}$ consist of all formulas in $\fl{\varphi}$ that
are of the form $\psi_1 \release \psi_2$ or $\nxt(\psi_1 \release \psi_2)$.
The notation is extended to formula sets in the usual way. 
The size $|\varphi|$ of a formula $\varphi$ is number of its subformulas.
Formulas are interpreted over paths $\pi$ of a transition systems 
$\Transsys = (\States,s^*,\to,\lambda)$. We have $\Transsys, \pi \models \true$ but not $\Transsys, \pi \models \false$
for any $\Transsys$ and $\pi$; and the semantics of the other constructs is given as follows.
\[\begin{array}{@{\Transsys, \pi \models\;}l@{\qquad\mbox{iff}\qquad}l}
  p                  & p \in \lambda(\pi(0)) \\
  \neg p             & p \notin \lambda(\pi(0)) \\
  \varphi \lor \psi  & \Transsys, \pi \models \varphi \mbox{ or } \Transsys, \pi \models \psi \\
  \varphi \land \psi & \Transsys, \pi \models \varphi \mbox{ and } \Transsys, \pi \models \psi \\
  \nxt\varphi &\Transsys, \pi^1 \models \varphi \\
  \varphi \until \psi & \exists k \in \Nat, \Transsys, \pi^k \models \psi  \mbox{ and } \forall j<k: \Transsys, \pi^j \models \varphi \\
  \varphi \release \psi & \forall k \in \Nat, \Transsys, \pi^k \models \psi  \mbox{ or } \exists j<k: \Transsys, \pi^j \models \varphi \\
  \expath\varphi & \exists \pi', \mbox{ s.t. } \pi'(0) = \pi(0) \mbox{ and } \Transsys, \pi' \models \varphi \\
  \allpath\varphi & \forall \pi', \mbox{ if } \pi'(0) = \pi(0) \mbox{ then } \Transsys, \pi' \models \varphi \\
\end{array}\]

Two formulas $\varphi$ and $\psi$ are equivalent, written $\varphi \equiv \psi$, if for all paths $\pi$
of all transition systems $\Transsys$: $\Transsys, \pi \models \varphi$ iff $\Transsys, \pi \models \psi$.

A formula $\varphi$ is called a \emph{state formula} if for all $\Transsys, \pi,\pi'$ with 
$\pi(0) = \pi'(0)$ we have $\Transsys, \pi \models \varphi$ iff $\Transsys, \pi' \models \varphi$. Hence,
satisfaction of a state formula in a path only depends on the first state of the path. Note that $\varphi$
is a state formula iff $\varphi \equiv\expath\varphi$. For state formulas we also write 
$\Transsys, s \models \varphi$ for $s \in \States$. \ctlstar is the set of all branching-time formulas which
are state formulas. A \ctlstar formula $\varphi$ is \emph{satisfiable} if there is a transition system 
$\Transsys$ with an initial state $s^*$ s.t.\ $\Transsys, s^* \models \varphi$. 

Finally, we introduce the two most well-known fragments of \ctlstar, namely \ctl and \ctlplus. In \ctl, no
Boolean combinations or nestings of temporal operators are allowed; they have to be immediately preceded by
a path quantifier. The syntax is given by the following grammar starting with $\varphi$.
\begin{align}
  \label{eq:ctl grammar:varphi}
  \varphi \quad &{::=} \quad \true \mid \false \mid p \mid  \neg p  \mid  \varphi \lor \varphi  \mid \varphi \land \varphi  \mid  \expath \psi \mid \allpath \psi 
\\
  \label{eq:ctl grammar:psi}
  \psi    \quad &{::=} \quad   \nxt \varphi     \mid  \varphi \until \varphi   \mid  \varphi \release \varphi
\end{align}
Formulas generated by $\varphi$ are state formulas.

The logic \ctlplus lifts the syntactic restriction slightly: it allows Boolean combinations of path operators inside
a path quantifier, but no nestings thereof. It is defined by the following grammar starting with $\varphi$.
\begin{align}
  \label{eq:ctlplus grammar:varphi}
  \varphi \quad &{::=} \quad \true \mid \false \mid p \mid  \neg p  \mid  \varphi \lor \varphi  \mid \varphi \land \varphi  \mid  \expath \psi \mid \allpath \psi 
\\
  \label{eq:ctlplus grammar:psi}
  \psi    \quad &{::=} \quad \varphi     \mid  \psi \lor \psi  \mid \psi \land \psi       \mid  \nxt \varphi     \mid  \varphi \until \varphi
\mid  \varphi \release \varphi
\end{align}
It should be clear that \ctl is a fragment of \ctlplus which is, in turn, a fragment of \ctlstar. However, only
the latter inclusion is proper w.r.t.\ expressivity as stated in the following.

\begin{proposition}[\cite{Emerson:1986:SNN}]
\ctlstar is strictly more expressive than \ctlplus, and \ctlplus is as expressive as \ctl.
\end{proposition}

Nevertheless, there are families of properties which can be expressed in \ctlplus using a family of formulas that 
is linearly growing in size, whereas every family of \ctl formulas expressing these properties must have exponential 
growth. This is called an exponential succinctness gap between the two logics.

\begin{proposition}[\cite{Wilke99,LISC01*197,ipl-ctlplus08}]
There is an exponential succinctness gap between \ctlplus and \ctl.
\end{proposition}

Such a succinctness gap can cause different complexities of decision procedures for the involved logics despite
equal expressive power. This is true in this case. 

\begin{proposition}[\cite{EmersonHalpern85,Fischer79}] 
Satisfiability checking for \ctl is \EXPTIME-complete.
\end{proposition}

The exponential succinctness gap causes on exponentially more difficult satisfiability problem which is shared with
that of the more expressive \ctlstar.

\begin{proposition}[\cite{Emerson:2000:CTA,STOC85*240,JL-icalp03}]
Satisfiability checking for \ctlstar and for \ctlplus are both 2\EXPTIME-complete.
\end{proposition}

\newcommand{\badtableauexample}{{
{
                                        \AxiomC{\tikzlabel{LoopLeft0}{$\allpath({\generally p},{\finally\generally p}), \allpath($\tikzlabel{bt15}{$\finally\generally p$}$)$}}
                                      \LeftLabel{\XruleNoE}
                                      \UnaryInfC{{$\allpath({\nxt \generally p},{\nxt \finally\generally p}), \allpath($\tikzlabel{bt14}{$\nxt \finally\generally p$}$)$}}
                                    \UnaryInfC{$ \allpath({\nxt \generally p},{\nxt \finally\generally p}), \allpath( p,$\tikzlabel{bt13}{$\nxt \finally\generally p$}$)$}
                                \UnaryInfC{$\allpath(\generally p,$ \tikzlabel{bt12}{$\nxt \finally\generally p$}$)$}
                              \doubleLine
                              \UnaryInfC{\tikzlabel{LoopLeft1}{$\allpath({\generally p},{\finally\generally p}), \allpath($\tikzlabel{bt11}{$\finally\generally p$}$)$}}
                            \LeftLabel{\XruleNoE}
                          \UnaryInfC{{$ \allpath({\nxt \generally p}, {\nxt \finally\generally p}), \allpath($\tikzlabel{bt10}{$\nxt \finally\generally p$}$), \neg p$}}
                        \UnaryInfC{$ \allpath({\nxt \generally p}, {\nxt \finally\generally p}), \allpath(p,$\tikzlabel{bt9}{$\nxt \finally\generally p$}$), \neg p$}
                    \doubleLine
                    \UnaryInfC{$ \allpath({\generally p},$ \tikzlabel{bt8}{$\nxt \finally\generally p$}$), \expath(\neg p)$}
                  \doubleLine
                  \UnaryInfC{$ \allpath({\generally p},$ \tikzlabel{bt7}{$\finally\generally p$}$), \expath(\finally\neg p)$}
                    \AxiomC{\tikzlabel{LoopRight0}{${\allpath({\generally p}, {\nxt \finally\generally p})}, {\expath({\expath\finally\neg p}, {\nxt \generally\expath\finally\neg p})}$}}
                  \doubleLine
                  \UnaryInfC{$ {\allpath({\generally p}, {\finally\generally p})}, {\expath({\generally\expath\finally\neg p})}$}
                \LeftLabel{\XruleWithE}
                \BinaryInfC{{$p, \allpath({\nxt \generally p},$ \tikzlabel{bt6}{$\nxt \finally\generally p$}$), \expath(\nxt \finally\neg p), {\expath({\nxt \generally\expath\finally\neg p})}$}}
            \doubleLine
            \UnaryInfC{$ {\allpath( p, \nxt \finally\generally p)}, \allpath({\nxt \generally p},$ \tikzlabel{bt5}{$\nxt \finally\generally p$}$), \expath(\finally\neg p), {\expath({\nxt \generally\expath\finally\neg p})}$}
      \doubleLine
      \UnaryInfC{\tikzlabel{LoopRight1}{$\allpath({\generally p},$ \tikzlabel{bt4}{$\nxt \finally\generally p$}$), {\expath(\expath\finally\neg p, {\nxt \generally\expath\finally\neg p})}$}}
    \doubleLine
    \UnaryInfC{$\allpath($\tikzlabel{bt3}{$\finally\generally p$}$), {\expath({\generally\expath\finally\neg p})}$}
  \doubleLine
  \UnaryInfC{$\expath($\tikzlabel{bt2}{$\allpath\finally\generally p$}$, \expath\generally\expath\finally\neg p)$}
  \UnaryInfC{$\expath($\tikzlabel{bt1}{${\allpath\finally\generally p} \wedge {\expath\generally\expath\finally\neg p})$}}
\DisplayProof

\tikzoverlay{
    \draw[draw, ->, thick, rounded corners=20pt] (LoopLeft0.west) to[myncbar,angle=180,arm=2cm] (LoopLeft1.west);%
	\draw[draw, ->, thick, rounded corners=20pt] (LoopRight0.east) to[myncbar,angle=0,arm=1cm] (LoopRight1.east);
	\tikzedge{bt1}{bt2}
	\tikzedge{bt2}{bt3}
	\tikzedge{bt3}{bt4}
	\tikzedge{bt4}{bt5}
	\tikzedge{bt5}{bt6}
	\tikzedge{bt6}{bt7}
	\tikzedge{bt7}{bt8}
	\tikzedge{bt8}{bt9}
	\tikzedge{bt9}{bt10}
	\tikzedge{bt10}{bt11}
	\tikzedge{bt11}{bt12}
	\tikzedge{bt12}{bt13}
	\tikzedge{bt13}{bt14}
	\tikzedge{bt14}{bt15}
	\tikzloopedge{bt15}{bt11}
}
}

}}

\newcommand{\goodtableauexample}{{
{
                                          \AxiomC{\tikzlabel{LoopLeft0}{$\allpath($\tikzlabel{gt16a}{$\generally p$}, $\nxt \finally\generally p {)}$}}
                                        \UnaryInfC{$\allpath($\tikzlabel{gt15a}{$\generally p$}, $\finally\generally p {)}$}
                                      \LeftLabel{\XruleNoE}
                                      \UnaryInfC{$p, \allpath($\tikzlabel{gt14a}{$\nxt \generally p$}, $\nxt \finally\generally p {)}$}
                                    \UnaryInfC{$ \allpath( p, \nxt \finally\generally p), \allpath($\tikzlabel{gt13a}{$\nxt \generally p$}, $\nxt \finally\generally p {)}$}
                                \UnaryInfC{\tikzlabel{LoopLeft1}{$\allpath($\tikzlabel{gt12a}{$\generally p$}, $\nxt \finally\generally p {)}$}}
                              \UnaryInfC{$\allpath(\finally\generally p), \allpath($\tikzlabel{gt11a}{$\generally p$}, $\finally\generally p {)}$}
                            \LeftLabel{\XruleNoE}
                          \UnaryInfC{$\allpath(\nxt\finally\generally p), \allpath($\tikzlabel{gt10a}{$\nxt\generally p$}, $\nxt\finally\generally p {)}, \neg p$}
                        \UnaryInfC{$\allpath(p, \nxt\finally\generally p), \allpath($\tikzlabel{gt9a}{$\nxt\generally p$}, $\nxt\finally\generally p {)}, \neg p$}
                        \doubleLine
                    \UnaryInfC{$\allpath($\tikzlabel{gt8a}{$\generally p$}, $\nxt\finally\generally p {)}, \expath(\neg p)$}
                  \doubleLine
                  \UnaryInfC{$\allpath($\tikzlabel{gt7a}{$\generally p$}, $\finally\generally p {)}, \expath(\finally\neg p)$}
                    \AxiomC{\tikzlabel{LoopRight0}{$\allpath($\tikzlabel{gt8b}{$\generally p$}, $\nxt\finally\generally p), {\expath(} \expath\finally\neg p$, \tikzlabel{gt8c}{$\nxt\generally\expath\finally\neg p$}$)$}}
                  \doubleLine
                  \UnaryInfC{$\allpath($\tikzlabel{gt7b}{$\generally p$}, $\finally\generally p), \expath($\tikzlabel{gt7c}{$\generally\expath\finally\neg p)$}}
                \LeftLabel{\XruleWithE}
                \BinaryInfC{$p, \allpath($\tikzlabel{gt6a}{$\nxt\generally p$}, $\nxt \finally\generally p {)}, \expath(\nxt \finally\neg p), \expath($\tikzlabel{gt6b}{$\nxt \generally\expath\finally\neg p)$}}
            \doubleLine
            \UnaryInfC{$\allpath(p, \nxt\finally\generally p), \allpath($\tikzlabel{gt5a}{$\nxt\generally p$}, $\nxt\finally\generally p {)}, \expath(\finally\neg p), \expath($\tikzlabel{gt5b}{$\nxt\generally\expath\finally\neg p)$}}
      \doubleLine
      \UnaryInfC{\tikzlabel{LoopRight1}{$\allpath($\tikzlabel{gt4a}{$\generally p$}, $\nxt\finally\generally p), \expath(\expath\finally\neg p,$ \tikzlabel{gt4b}{$\nxt\generally\expath\finally\neg p$}$)$}}
    \doubleLine
    \UnaryInfC{$\allpath($\tikzlabel{gt3a}{$\finally\generally p),$} $\expath($\tikzlabel{gt3b}{$\generally\expath\finally\neg p)$}}
  \doubleLine
  \UnaryInfC{$\expath($\tikzlabel{gt2a}{$\allpath\finally\generally p,$} \tikzlabel{gt2b}{$\expath\generally\expath\finally\neg p)$}}
  \UnaryInfC{$\expath($\tikzlabel{gt1}{$\allpath\finally\generally p \wedge \expath\generally\expath\finally\neg p)$}}
\DisplayProof

\tikzoverlay{
    \draw[draw, ->, thick, rounded corners=20pt] (LoopLeft0.west) to[myncbar,angle=180,arm=2cm] (LoopLeft1.west);%
    \draw[draw, ->, thick, rounded corners=20pt] (LoopRight0.east) to[myncbar,angle=0,arm=1cm] (LoopRight1.east);
	\tikzedge{gt1}{gt2a}
	\tikzedge{gt1}{gt2b}
	\tikzedge{gt2a}{gt3a}
	\tikzedge{gt2b}{gt3b}
	\tikzedge{gt3a}{gt4a}
	\tikzedge{gt3b}{gt4b}
	\tikzedge{gt4a}{gt5a}
	\tikzedge{gt4b}{gt5b}
	\tikzedge{gt5a}{gt6a}
	\tikzedge{gt5b}{gt6b}
	\tikzedge{gt6a}{gt7a}
	\tikzedge{gt6a}{gt7b}
	\tikzedge{gt6b}{gt7c}
	\tikzedge{gt7a}{gt8a}
	\tikzedge{gt7b}{gt8b}
	\tikzedge{gt7c}{gt8c}
	\tikzloopedgeleft{gt8b}{gt4a}
	\tikzloopedge{gt8c}{gt4b}
	\tikzedge{gt8a}{gt9a}
	\tikzedge{gt9a}{gt10a}
	\tikzedge{gt10a}{gt11a}
	\tikzedge{gt11a}{gt12a}
	\tikzedge{gt12a}{gt13a}
	\tikzedge{gt13a}{gt14a}
	\tikzedge{gt14a}{gt15a}
	\tikzedge{gt15a}{gt16a}
	\tikzloopedgeleft{gt16a}{gt12a}
}
}

}}

\section{Satisfiability Games for \ctlstar}
\label{sec:tableaux}

Here we are concerned with special 2-player zero-sum games of perfect information. They can be seen as a finite, directed
graph whose node set is partitioned into sets belonging to each player. Formally, a game is a tuple $\mathcal{G} = (V,V_0,E,v_0,L)$
where $(V,E)$ is a directed graph. We restrict our attention to total graphs, i.e.\ every node is assumed to have at least
one successor. The set $V_0 \subseteq V$ consists of all nodes owned by player 0. This naturally induces the set $V_1 := V \setminus V_0$
of all nodes owned by player 1. The node $v_0$ is the designated initial node. 

Any play starts from this initial node by placing a token there. Whenever the token is on a node that belongs to player $i$, 
it is his/her turn to push it along an edge to a successor node. In the infinite, this results in a play, and the winning 
condition $L \subseteq V^\omega$ prescribes which of these plays are won by player 0.

A \emph{strategy} for player $i$ is a function $\sigma: V^*V_i \to V$ which tells him/her how to move in any given 
situation in a play. Formally, a play $v_0,v_1,\ldots$ conforms
to strategy $\sigma$ for player $i$, if for every $j$ with $v_j \in V_i$ we have $v_{j+1} = \sigma(v_0 \ldots v_j)$.
A \emph{winning strategy} for player $i$ is a strategy such that he/she wins any play regardless of the opponent's choices. Formally, $\sigma$
is a winning strategy if for all plays $\pi = v_0,v_1,\ldots$ that conform to $\sigma$ we have $\pi \in L$.

It is easy to relax the requirements of totality. In that case we attach two
additional designated nodes $\mathsf{win}_0$ and $\mathsf{win}_1$ such that
every node originally without successors gets an edge to either of them, each of
these only has one edges to itself, and the winning condition $L$ includes all words of the form $V^* \mathsf{win}_0^\omega$ and excludes all of the form $V^*\mathsf{win}_1^\omega$.
In the following we will therefore allow ourselves to have plays ending in states without successors which can be turned into total
games using this simple transformation. In other words every dead end is lost by the player who owns the node.

\subsection{The Game Rules}
\label{subsec:rules}

We present satisfiability games for branching-time state formulas in negation normal form. Let $\vartheta$ be
a \ctlstar-formula fixed for the remainder of this section. For convenience, the games will be presented w.r.t.\ to 
this particular formula $\vartheta$.

We need the following notions: $\Sigma$ and $\Pi$ are finite (possibly empty) sets of formulas with $\Sigma$ 
being interpreted as a disjunction of formulas and $\Pi$ as a conjunction. A \emph{quantifier-bound formula block} 
is an $\expath$- or $\allpath$-labelled set of formulas, i.e.\ either a $\expath\Pi$ or a $\allpath\Sigma$. 
Any set under an $\expath$-bound resp.\ $\allpath$-bound block is assumed to be
read as a conjunction resp.\ disjunction of the formulas.
We identify 
an empty $\Sigma$ with $\false$ and an empty $\Pi$ with $\true$. We write 
$\Lambda$ for a set of literals. For a set of formulas $\Gamma$ let $\nxt\Gamma := \{ \nxt\psi \mid \psi \in \Gamma \}$.

In order to ease readability we will omit as many curly brackets as possible and often use round brackets to group formulas
into a set. For instance $\expath(\varphi \wedge \psi, \Pi)$ denotes a block that is prefixed by $\expath$ and which consists
of the union of $\Pi$ and $\{\varphi \wedge \psi\}$, implicitly assuming that this does not occur in $\Pi$ already.

A \emph{configuration (for $\vartheta$)} is a non-empty set of the form
\[
\allpath\Sigma_1, \ldots, \allpath\Sigma_n, \expath\Pi_1, \ldots, \expath\Pi_m, \Lambda
\]
where $n,m \ge 0$, and $\Sigma_1,\ldots,\Sigma_n,\Pi_1,\ldots,\Pi_m,\Lambda$ are subsets of
$\fl{\vartheta}$. The meaning of such a configuration is given by the state formula
\[
\bigwedge_{i=1}^n \allpath\big( \hspace*{-1mm}\bigvee_{\psi \in \Sigma_i}\hspace*{-1mm} \psi\big)  \wedge 
\bigwedge_{i=1}^m \expath\big(\hspace*{-1mm} \bigwedge_{\psi \in \Pi_i}\hspace*{-1mm} \psi\big)  \wedge 
\bigwedge_{\ell \in \Lambda} \ell\ .
\]
Note that a configuration only contains existentially quantified conjunctions and universally quantified disjunctions as
blocks. There are no blocks of the form $\expath\Sigma$ or $\allpath\Pi$ simply because an existential path quantifier
commutes with a disjunction, and so does a universal path quantifier with a conjunction. Thus, $\allpath\Sigma$ would
be equivalent to $\bigwedge \{ \allpath\varphi \mid \varphi \in \Sigma \}$ for instance.

A configuration $\mathcal C$ is \emph{consistent} if it does not contain $\false$ and there is no $p \in \Prop$ s.t.\ 
$p \in \mathcal C$ and $\neg p \in \mathcal C$. Note that the meaning of an inconsistent configuration is unsatisfiable, 
but the converse does not hold because consistency is only concerned with the occurrence of literals. Unsatisfiability 
can also be given be conflicting temporal operators, e.g.\ $\expath(\nxt p, \nxt \neg p)$.
 
We write $\goals{\vartheta}$ for the set of all consistent configurations for $\vartheta$. Note that this is a
finite set of at most doubly exponential size in $|\vartheta|$.

\begin{figure}[t]%
\[
  \unaryRule{\AAndrule}
    {\allpath(\varphi, \Sigma), \allpath(\psi, \Sigma)}
    {\allpath(\varphi \land \psi, \Sigma)}
\qquad
  \unaryRule{\AOrrule}
    {\allpath(\varphi, \psi, \Sigma)}
    {\allpath(\varphi \lor \psi, \Sigma)}
\qquad
  \choiceRule{\ALitrule}
    {\ell}
    {\allpath\Sigma}
    {\allpath(\ell, \Sigma)}
\]
\[
  \unaryRule{\AUrule}
    {\allpath(\psi, \varphi, \Sigma), \allpath(\psi, \nxt(\varphi \until \psi), \Sigma)}
    {\allpath(\varphi \until \psi, \Sigma)}
\qquad
  \unaryRule{\ARrule}
    {\allpath(\psi, \Sigma), \allpath(\varphi, \nxt(\varphi \release \psi), \Sigma)}
    {\allpath(\varphi \release \psi, \Sigma)}
\]
\[
  \choiceRule{\AArule}
    {\allpath\varphi}
    {\allpath\Sigma}
    {\allpath(\allpath\varphi, \Sigma)}
\qquad
  \choiceRule{\AErule}
    {\expath\varphi}
    {\allpath\Sigma}
    {\allpath(\expath\varphi, \Sigma)}
\qquad
  \unaryRule{\Ettrule}
    {}
    {\expath\true}
\]
\[
  \choiceRule{\EOrrule}
    {\expath(\varphi,\Pi)}
    {\expath(\psi,\Pi)}
    {\expath(\varphi \lor \psi,\Pi)}
\qquad
  \unaryRule{\EAndrule}
    {\expath(\varphi, \psi, \Pi)}
    {\expath(\varphi \land \psi, \Pi)}
\qquad
  \unaryRule{\ELitrule}
    {\expath\Pi,\ell}
    {\expath(\ell,\Pi)}
\]
\[
  \choiceRule{\EUrule}
    {\expath(\psi,\Pi)}
    {\expath(\varphi,\nxt(\varphi \until \psi),\Pi)}
    {\expath(\varphi \until \psi,\Pi)}
\qquad
  \unaryRule{\EErule}
    {\expath\varphi, \expath\Pi}
    {\expath(\expath\varphi,\Pi)}
\]
\[
  \choiceRule{\ERrule}
    {\expath(\psi,\varphi,\Pi)}
    {\expath(\psi,\nxt(\varphi \release \psi),\Pi)}
    {\expath(\varphi \release \psi,\Pi)}
\qquad
  \unaryRule{\EArule}
    {\allpath\varphi, \expath\Pi}
    {\expath(\allpath\varphi,\Pi)}
\]
\[
  \def\addSideFormulas{}
  \unaryRule{\XruleNoE}
    {\allpath\Sigma_1, \ldots, \allpath\Sigma_m}
    {\allpath\nxt\Sigma_1, \ldots, \allpath\nxt\Sigma_m, \Lambda}
\qquad
  \branchingRule{\XruleWithE}
    {\expath\Pi_1, \allpath\Sigma_1, \ldots, \allpath\Sigma_m}
    {\expath\Pi_n, \allpath\Sigma_1, \ldots, \allpath\Sigma_m}
    {\expath\nxt\Pi_1, \ldots, \expath\nxt\Pi_n, \allpath\nxt\Sigma_1, \ldots, \allpath\nxt\Sigma_m, \Lambda}
\]

\caption{The game rules for \ctlstar.}
\label{fig:pretableaurules}
\end{figure}

\begin{definition}\label{def:sat game} 
The satisfiability game $\mathcal{G}_{\vartheta}$ for the formula $\vartheta$ is a game $(\goals{\vartheta},V_0,E,v_0,L)$
whose nodes are all possible configurations and whose edge relation is given by the game rules in Figure~\ref{fig:pretableaurules}.
They also determine which configurations belong to player 0, i.e.\ to $V_0$, namely all but those to which rule $\XruleWithE$ 
is applied. 

Note that the rules are written such that a configuration at the bottom of the rule has, as its successors, all configurations
at the top of the rule. It is only rules $\ALitrule$, $\AArule$, $\AErule$, $\EOrrule$, $\EUrule$, and $\ERrule$ which always produce two
successors, rule $\XruleWithE$ can have an arbitrary number of successors that
is at least one.
It is understood that the formulas which are stated explicitly under the line do not occur 
in the sets~$\Lambda$ or~$\sideFormulas$. The symbol $\ell$ stands for an arbitrary literal.

The initial configuration is $v_0 = \expath\vartheta$. The winning condition $L$ will be described 
in Definition~\ref{def:G winning condition} in the next subsection.
\end{definition}

As for the representation, examples in this paper will use tailored rules for the abbreviations $\finally$ and 
$\generally$ instead of the rules \AUrule{}, \EUrule{}, \ARrule{} and \ERrule{}. Take for instance a construct of the form
$\allpath\finally\psi$. A rule for this can easily be derived by applying the rules for the unabbreviated version of this.
\[
  \AxiomC{$\allpath(\psi, \nxt\finally\psi, \Sigma), \sideFormulas, \true$}
  \LeftLabel{\ALitrule}
  \UnaryInfC{$\allpath(\psi, \true, \Sigma), \allpath(\psi, \nxt\finally\psi, \Sigma), \sideFormulas$}
  \LeftLabel{\AUrule}
  \UnaryInfC{$\allpath(\,\underbrace{\true \until \psi}_{\enspace=\enspace \finally\psi}\,, \Sigma), \sideFormulas$}
  \bottomAlignProof\DisplayProof
\]
The additional $\true$ in the literal part does not affect 
consistency of a configuration, nor the applicability of any other rule. Hence, it can be dropped. 
Therefore, we can use the following rule for this construct.
\[
  \unaryRule{\AFrule}
    {\allpath(\psi, \nxt\finally\psi,\Sigma)}
    {\allpath(\finally\psi,\Sigma)}
\]
The other abbreviated rules follow the same lines.
\[
  \unaryRule{\AGrule}
    {\allpath(\psi, \Sigma), \allpath(\nxt\generally\psi, \Sigma)}
    {\allpath(\generally \psi, \Sigma)}
  \qquad 
  \choiceRule{\EFrule}
    {\expath(\psi,\Pi)}
    {\expath(\nxt\finally\psi,\Pi)}
    {\expath(\finally \psi,\Pi)}
\]
\[
  \unaryRule{\EGrule}
    {\expath(\psi,\nxt\generally\psi,\Pi)}
    {\expath(\generally \psi,\Pi)}
\]
Note that for the $\EGrule$ rule --- which is based on the $\ERrule$ rule --- it is
never the wrong choice to select the right alternative instead of the left one.
Choosing the left one would leave us with a configuration denoting $\expath(\psi\land\false\land\bigwedge\Pi)\land\bigwedge\sideFormulas$
which can never be satisfied because of the constant~$\false$.

\begin{example}
\label{ex:pre-tableau}
A strategy for player 0 in the game on $\allpath\finally\generally p \; \wedge \; \expath\generally\expath\finally\neg p$ is represented
in Figure~\ref{fig: bad tableau}. Note that such strategies can be seen as infinite trees. The bold arrows in 
Figure~\ref{fig: bad tableau} point towards repeating configurations in this strategy. This is meant to represent the infinite tree that
is obtained by repeatedly continuing as it is done in the two finite branches. Also note that in general, strategies may not
be representable in a finite way like this. 
The twin lines indicate hidden configurations whenever unary rules can be applied in parallel.
For instance, the double line at the bottom represents the application of the rules \AFrule{} and \EGrule{}.
The thin arrows will only be used in the next subsection in order to explain the
winning conditions in the satisfiability games.
A strategy for player 0 induces canonically a tree model by collapsing successive configurations that are not separated by applications
of the rules $\XruleNoE$ and $\XruleWithE$. Doing this to the strategy in Figure~\ref{fig: bad tableau} results in the
following transition system. Note that the tableau of Figure~\ref{fig: bad
tableau} gives no specification on whether $p$ should be included in the
right-most node. It is natural to only include those propositions that are
required to be true.
\begin{center}
\includegraphics{./examplesystems.1}
\end{center}
Note that it does not satisfy the formula $\allpath\finally\generally p \wedge \expath\generally\expath\finally\neg p$. The overall
goal is to characterise satisfiability in \ctlstar through these games. Hence, it is important to define the winning conditions 
such that this strategy is not a winning strategy. 
\end{example}

\begin{figure}[t]
\begin{center}
\includegraphics{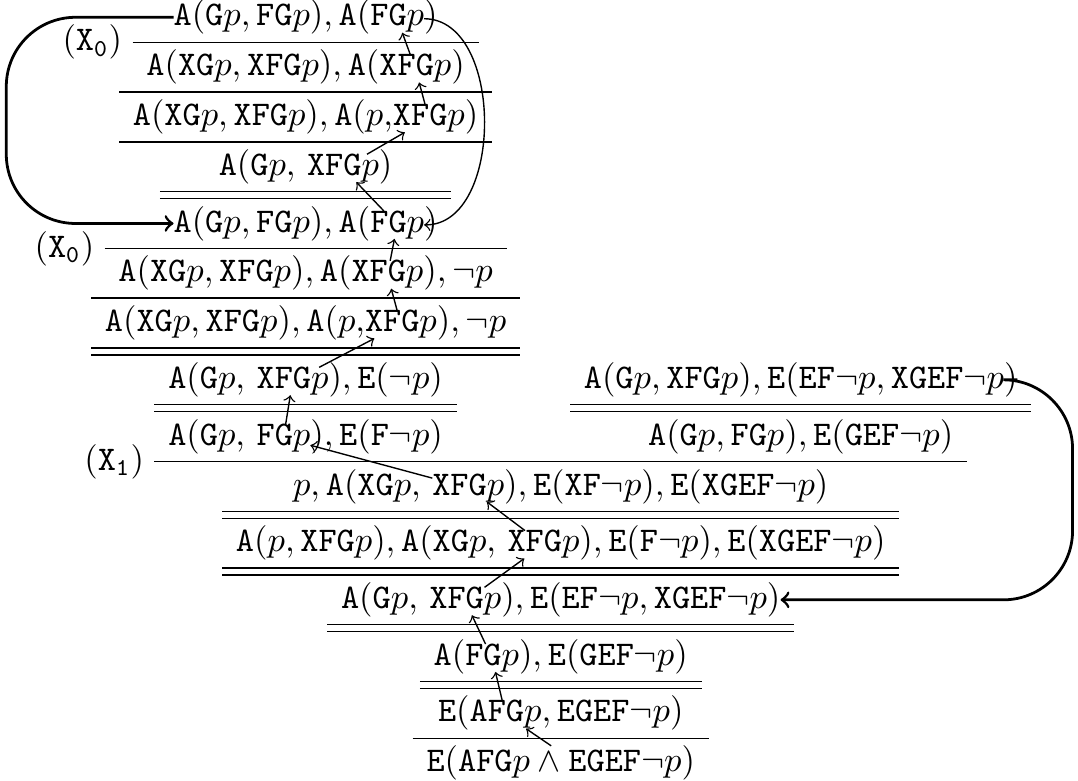}
\end{center}
\caption{A strategy for player 0 in the satisfiability game for $\allpath\finally\generally p \wedge \expath\generally\expath\finally\neg p$.}
\label{fig: bad tableau}
\end{figure}

\subsection{The Winning Conditions}
\label{sub:winning:cond}

An occurrence of a formula is called \emph{principal} if it gets transformed by a rule. For example, 
the occurrence of $\varphi \wedge \psi$ is principal in $\EAndrule$. A principal formula has
\emph{descendants} in the successor configurations. For example, both occurrences of $\varphi$ and $\psi$ are 
descendants of the principal $\varphi \wedge \psi$ in rule $\EAndrule$.

Note that in the modal rules 
$\XruleNoE$ and $\XruleWithE$, every formula apart from those in the literal part is principal. 
Literals in the literal part can never be principal, but literals inside of an $\allpath$- or $\expath$-block
are principal in rules $\ALitrule$ and $\ELitrule$. 
Finally, any non-principal occurrence of a formula in a configuration may have a \emph{copy} in one of the successor configurations.
The copy is the same formula since it has not been transformed. For instance, any formula in $\Sigma$ in 
rule $\ALitrule$ has a copy in the successor written on the right but does not have a copy in the successor
on the left.

The gap between the existence of strategies for player 0 and satisfiability is caused by unfulfilled eventualities:
an eventuality is a formula of the form $\until$ or its abbreviation $\finally$. Note how the rules handle these by
unfolding using the \ctlstar equivalence $Q(\psi_1 \until \psi_2) \equiv Q(\psi_2 \vee (\psi_1 \wedge \nxt(\psi_1 \until \psi_2)))$
for any $Q \in \{\expath,\allpath\}$. The rules for the Boolean operators and for the $\nxt$-modalities can lead to a 
configuration in which $\psi_1 \until \psi_2$ occurs again inside of a $Q$-block. Note that inside an $\expath$-block
this is only possible if player 0 decides not to choose the successor containing $\psi_2$. Inside of an $\allpath$-block
the situation is slightly different; player 0 has no choices there. Still, it is important to note that a $\until$-formula
should not be unfolded infinitely often because $\psi_1 \until \psi_2$ asserts that eventually $\psi_2$ will be true, and
unfolding postpones this by one state in a possible model. Thus, the winning conditions have to ensure that player 0 cannot
let an eventuality formula get unfolded infinitely many times without its right argument being satisfied infinitely many
times as well.

In order to track the infinite behaviour of eventualities, one needs to follow single
formulas through the branches that get transformed by a rule from time to time. Note that 
a formula can occur inside of several blocks. Thus it is important to keep track of the block
structure as well.

In the following we develop the technical definitions that are necessary in order to capture such unfulfilled 
eventualities and present some of their properties.

\begin{definition}
A quantifier-bound block $\allpath\Sigma$ or $\expath\Pi$ is called \emph{principal} as well if it
contains a principal formula. A quantifier-bound block might have \emph{descendants} in the successor(s). For example,
$\allpath(\varphi \land \psi,\Sigma)$ has two descendants $\allpath(\varphi,\Sigma)$ and
$\allpath(\psi,\Sigma)$ in an application of $\AAndrule$.
\end{definition}

\begin{definition}
Let $\mathcal C_1$ be a configuration to which a rule $r$ is applicable and
let $\mathcal C_2$ be one of its successors. Furthermore, let $\qpath_1\Delta_1$, resp.\ 
  $\qpath_2\Delta_2$ with $\qpath_1,\qpath_2 \in \{\expath,\allpath\}$ and $\Delta_1,\Delta_2 \subseteq
  \fl{\vartheta}$ be quantifier-bound blocks occurring in the $\allpath$- or $\expath$-part of
  $\mathcal C_1$, resp.\ $\mathcal C_2$.  We say that $\qpath_1\Delta_1$ is \emph{connected} to
  $\qpath_2\Delta_2$ in $\mathcal C_1$ and $\mathcal C_2$, if either
\begin{itemize}[leftmargin=1em]
\item $\qpath_1\Delta_1$ is principal in $r$, and $\qpath_2\Delta_2$ is one of its descendants in $\mathcal C_2$; or
\item $\qpath_1\Delta_1$ is not principal in $r$ and $\qpath_2\Delta_2$ is a copy of $\qpath_1\Delta_1$ in $\mathcal C_2$.
\end{itemize}
We write this as $(\mathcal C_1,\qpath_1\Delta_1) \leadsto (\mathcal C_2,\qpath_2\Delta_2)$. If the rule
instance can be inferred from the context we may also simply write $\qpath_1\Delta_1 \leadsto
\qpath_2\Delta_2$.  Additionally, let $\psi_1$, resp.\ $\psi_2$ be a formula occurring in $\Delta_1$,
resp.\ $\Delta_2$.  We say that $\psi_1$ is \emph{connected} to $\psi_2$ in
$(\mathcal C_1 ,\qpath_1\Delta_1)$ and $(\mathcal C_2,\qpath_2\Delta_2)$, if either
\begin{itemize}[leftmargin=1em]
\item $\psi_1$ is principal in $r$, and $\psi_2$ is one of its descendants in $\mathcal C_2$; or
\item $\psi_1$ is not principal in $r$ and $\psi_2$ is a copy of $\psi_1$ in $\mathcal C_2$.
\end{itemize}
We write this as $(\mathcal C_1,\qpath_1\Delta_1,\psi_1) \leadsto (\mathcal C_2,\qpath_2\Delta_2,\psi_2)$.
If the rule instance can be inferred from the context we may also simply write $(\qpath_1\Delta_1,\psi_1)
\leadsto (\qpath_2\Delta_2,\psi_2)$.
A block connection $(\mathcal C_1, \qpath_1\Delta_1) \leadsto
(\mathcal C_2, \qpath_2\Delta_2)$ is called \emph{spawning} if there is a formula $\psi$ s.t.\
$\qpath_2 \psi \in \Delta_1$ is principal and $\Delta_2 = \{\psi\}$.
The only rules that possibly induce a spawning block connection are $\EErule$, $\EArule$,
$\AArule$ and $\AErule$. For example $(C_1,
\allpath\{q,\expath p\}) \leadsto (C_2, \expath \{p\})$ is spawning while $(C_1,
\allpath\{q,\expath p\}) \leadsto (C_2, \allpath \{q\})$ is not.
\end{definition}

\begin{definition}\label{def:trace}
Let $\mathcal{C}_0,\mathcal{C}_1,\ldots$ be an infinite play of a satisfiability game for some formula $\vartheta$. 
A \emph{trace} $\Xi$ in this play is an infinite sequence $\qpath_0\Delta_0,\qpath_1\Delta_1,\ldots$ s.t.\
for all $i \in \Nat$: $(\mathcal{C}_i,\qpath_i\Delta_i) \leadsto (\mathcal{C}_{i+1},\qpath_{i+1}\Delta_{i+1})$.
A trace $\Xi$ is called an \emph{$\expath$-trace}, resp.\ \emph{$\allpath$-trace} if there is an $i \in \Nat$ s.t.\
$\qpath_j = \expath$, resp.\ $\qpath_j = \allpath$ for all $j \ge i$.
We say that a trace is \emph{finitely spawning} if it contains only finitely many spawning block connections.
\end{definition}

\begin{lemma} 
\label{modal rule infinitely often}
Every infinite play contains infinitely many applications of rules \XruleNoE{} or \XruleWithE{}.
\end{lemma}

\begin{proof}
First, we define the \emph{duration of a formula} $\psi$ as the syntactic height
when $\nxt$-subformulas are treated as atoms. More formally:
\begin{displaymath}
\mathrm{dur}: \psi \mapsto \begin{cases}
1 & \text{ if } \psi \equiv \true, \false, p, \neg p, \nxt\psi' \\
1 + \mathrm{dur}(\psi') & \text{ if } \psi \equiv \expath \psi', \allpath \psi' \\
1 + \max(\mathrm{dur}(\psi_1),\mathrm{dur}(\psi_2)) & \text{ if } \psi \equiv \psi_1 \lor \psi_2, \psi_1 \land \psi_2, \psi_1 \until \psi_2, \psi_1 \release \psi_2
\end{cases}
\end{displaymath}
A well-ordering $<$ on the duration of formulas is induced by the well-ordering on
natural numbers. Let $F$ be $\{\mathrm{dur}(\varphi) \mid \varphi \in \fl{\vartheta}\}$, the range of these durations,
and let $B$ be the range of all block sizes, that is $\{0, \ldots, |\fl{\vartheta}|\}$. Both sets are finite.

Second, we define the \emph{duration of a block} $\qpath \Delta$ as a map 
$\mathrm{dur}(\qpath \Delta): F \rightarrow B$ that returns the number of
subformulas of a certain duration. More formally:
\begin{displaymath}
\mathrm{dur}(\qpath \Delta): n \mapsto |\Delta \cap \mathrm{dur}^{-1}(n)|
\end{displaymath}
A well-ordering $\prec$ on the duration of blocks is given as follows (as the domain of the duration is finite and its range is well-founded).
\begin{displaymath}
f \prec g
 \quad:\iff\quad \exists n\in F: f(n) < g(n) \; \wedge \; \forall m > n: f(m) = g(m)
\end{displaymath}

Third, we define the \emph{duration of a configuration} $\mathcal{C}$ as a map
$\mathrm{dur}(\mathcal{C}): B^F \rightarrow \Nat$ that returns the number of
blocks of a certain duration. More formally:
\begin{displaymath}
\mathrm{dur}(\mathcal{C}): f \mapsto |\mathcal{C} \cap \mathrm{dur}^{-1}(f)|
\end{displaymath} 
A well-ordering $\lhd$ on the duration of configurations is given as follows.
\begin{displaymath}
C \lhd D
 \quad:\iff\quad \exists f\in B^F: C(f) < D(f) \; \wedge \; \forall g \succ f: C(g) = D(g)
\end{displaymath}
Indeed, $\lhd$ is well-founded as the domain of durations, $B^F$, is finite.

The claim now follows from the fact that every rule application except for
$\XruleNoE$ and $\XruleWithE$ strictly decreases the duration of the configuration.
\end{proof}

\begin{definition}\label{def:thread}
Let $\mathcal{C}_0,\mathcal{C}_1,\ldots$ be an infinite play. 
A \emph{thread} $t$ in a trace $\Xi=\qpath_0\Delta_0,\qpath_1\Delta_1,\ldots$ within $\mathcal{C}_0,\mathcal{C}_1,\ldots$
is an infinite sequence $\psi_0,\psi_1,\ldots$ s.t.\ for all $i \in \Nat$:
$(\mathcal{C}_i,\qpath_i\Delta_i,\psi_i) \leadsto (\mathcal{C}_{i+1},\qpath_{i+1}\Delta_{i+1},\psi_{i+1})$.
Such a thread $t$ is called a $\until$-thread, resp.\ $\release$-thread if there is a formula
$\varphi \until \psi \in \fl{\vartheta}$, resp.\ $\varphi \release \psi \in \fl{\vartheta}$ s.t.\
$\psi_j = \varphi \until \psi$, resp.\ $\psi_j = \varphi \release \psi$ for infinitely many~$j$.

An $\expath$-trace is called \emph{good} iff it has no $\until$-thread; similarly, an $\allpath$-trace
is called \emph{good} iff it has an $\release$-thread. In other words, an $\expath$-trace is called \emph{bad} if it 
contains an $\until$-thread, and an $\allpath$-trace is called \emph{bad} if it contains no $\release$-thread.
\end{definition}

\begin{lemma}\label{every trace form}
Every trace in an infinite play is either an $\allpath$-trace or an $\expath$-trace, and is only finitely spawning.
\end{lemma}

\begin{proof}
  Let $\qpath_0\Delta_0,\qpath_1\Delta_1,\ldots$ be a trace and assume that $\{i \mid
  \qpath_i\Delta_i \leadsto \qpath_{i+1}\Delta_{i+1} \text{ is spawning}\}$ is infinite. Let $i_0 < i_1 <
  \ldots$ be the ascending sequence of numbers in this infinite set and let $\phi_{i_j}$ denote the
  formula in the singleton set $\Delta_{i_j+1}$. Note that for all $j$ it is the case that
  $\phi_{i_{j+1}}$ is a proper subformula of $\phi_{i_j}$. Hence the set cannot be
  infinite. Now note that every finitely spawning trace eventually must be either an $\allpath$- or an
  $\expath$-trace because the change of the quantifier on the current block in a trace is only possible in a
  moment that the trace is spawning.
\end{proof}

\begin{lemma}\label{every thread form}
Every thread in a trace of an infinite play is either a $\until$- or an $\release$-thread.
\end{lemma}

\begin{proof}
Let $t = \psi_0,\psi_1,\ldots$ be a thread.
Assume that $t$ is neither a $\until$- nor an $\release$-thread, hence there is a position $i^*$ s.t.\
$\psi_i$ is neither of the form $\psi'\until\psi''$ nor of the form $\psi'\release\psi''$ for all $i \geq i^*$,
hence $\psi_{i+1}$ is a subformula of $\psi_{i}$ for all $i \geq i^*$. By Lemma~\ref{modal rule infinitely often}
it follows that $\psi_{i+1} \not= \psi_i$ for infinitely many $i$ which cannot be the case, hence $t$
has to be a $\until$- or an $\release$-thread.
Finally, assume that $t$ is both a $\until$- and an $\release$-thread, i.e.\ there are positions $i_0 <
i_1 < i_2$ s.t.\ $\psi_{i_0} = \psi_{i_2} = \psi'\release\psi''$ and $\psi_{i_1} = \varphi'\until\varphi''$. Hence
$\psi_{i_1}$ is a proper subformula of $\psi_{i_0}$ and $\psi_{i_2}$ is a proper subformula of $\psi_{i_1}$, 
thus $\psi_{i_0}$ would be a proper subformula of itself.
\end{proof}

\begin{lemma}\label{every UR-thread form}
For every $\until$- and every $\release$-thread $\psi_0, \psi_1, \ldots$ in a trace of an infinite play 
there is an $i \in \Nat$ such that $\psi_i$ is a $\until$-, or an $\release$-formula resp., and
$\psi_j = \psi_i$ or $\psi_j = \nxt \psi_i$ for all $j \geq i$.
\end{lemma}

\begin{proof}
  For all $i \in \Nat$, it holds that $\psi_{i+1}$ is a subformula $\psi_i$, or $\psi_{i+1} = \nxt \psi_i$
  provided that $\psi_i$ is a $\until$- or an $\release$-formula.  The map which removes the 
  frontal $\nxt$ from a formula converts the thread into a chain which is weakly decreasing with respect to the
  subformula order. Because this order is well-founded, the chain is eventually constant, say from $n$ onwards.
  By Lemma~\ref{modal rule infinitely often}, either \XruleNoE{} or \XruleWithE{} has been applied at a position $i-1$
  for some $i > n$.  Hence, $\psi_i$ is either a $\until$- or an $\release$-formula, and $i$ meets the claimed property.
\end{proof}

We now have obtained all the necessary technical material that is needed to define the winning conditions in 
the satisfiability game $\mathcal{G}_\vartheta$.

\begin{definition}\label{def:G winning condition}
The winning condition $L$ of $\mathcal{G}_{\vartheta} = (\goals{\vartheta},V_0,E,v_0,L)$ 
consists of every finite play which ends in a consistent set of literals,
and of every infinite play which does not contain a bad trace.
\end{definition}

In other words, player 0's objective is to create a play in which every $\until$-formula inside of an 
$\expath$-trace gets fulfilled eventually. She can control this using rule $\EUrule$. Inside of an $\allpath$-trace
She must hope that not every formula that gets unfolded infinitely often is of the $\until$-type. Note that sets 
inside of an $\expath$-block are conjunctions, hence, one unfulfilled formula makes the entire block false. Inside of
an $\allpath$-block the sets are disjunctions though, hence, in order to make this block true it suffices to satisfy
one of the formula therein. An $\release$-formula that gets unfolded infinitely often is---unlike an 
$\until$-formula---indeed satisfied.

\begin{figure}[t]
\begin{center}
\includegraphics{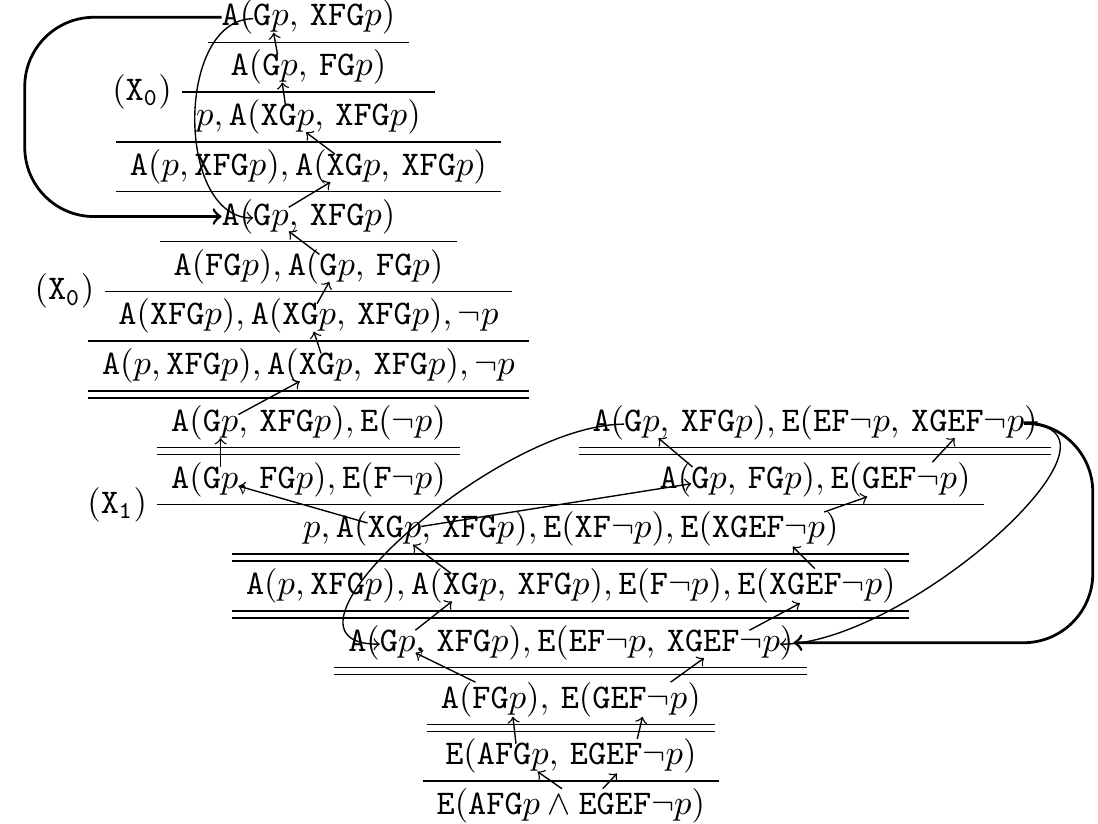}
\end{center}
\caption{A winning strategy for player 0 in the satisfiability game on $\allpath\finally\generally p \wedge \expath\generally\expath\finally\neg p$.}
\label{fig: good tableau}
\end{figure}

\begin{example}
Consider the strategy in Figure~\ref{fig: bad tableau} again. It is not a winning strategy because its left branch contains a bad
$\allpath$-trace, i.e.\ the eventuality $\finally\generally p$ is postponed for an infinite number of steps, which is
the only thread contained in the trace. Since this thread is an $\until$-thread, there is no $\release$-thread contained in the trace.

Figure~\ref{fig: good tableau} shows a winning strategy for player 0 in the game on this formula 
$\allpath\finally\generally p \; \wedge \; \expath\generally\expath\finally\neg p$. Infinite threads are 
being depicted using thin arrows. It is not hard to see that every $\allpath$-trace
contains a $\release$-thread and that every $\expath$-trace only contains $\release$-threads.
Again, this strategy induces a canonic model, but this time a satisfying one because it is in fact a winning strategy:
\begin{center}
\includegraphics{./examplesystems.2}
\end{center}
Note that in this case, all paths starting in the leftmost state will eventually only visit states that satisfy $p$.
Furthermore, there is a path---namely the loop on this state---on which every state is the beginning of a 
path---namely the one moving over to the right---on which $\neg p$ holds at some point.
\end{example}

Winning strategies, as opposed to ordinary strategies, exactly characterise
satisfiability of \ctlstar-formulas in the following sense.

\begin{theorem}
\label{thm:correctness}
For all $\vartheta \in$ \ctlstar: $\vartheta$ is satisfiable iff player 0 has a winning strategy for the satisfiability game
$\mathcal{G}_{\vartheta}$.
\end{theorem}

The proof is given in the following section.

\section{Correctness Proofs}
\label{sec:correctness}

This section contains the proof of Theorem~\ref{thm:correctness};
both implications -- soundness and completeness -- are considered separately.
The completeness proof is technically tedious but does not use any heavy machinery once the right invariants
are found. Given a model for $\vartheta$ we use these invariants to construct a winning strategy 
for player~$0$ in a certain way. 
Soundness can be shown by collapsing a winning strategy into a tree-like transition system and verifying that
it is indeed a model of $\vartheta$.

\subsection{Soundness}
\begin{theorem}%
  \label{thm:soundness}
  Suppose that player 0 has a winning strategy 
  for the satisfiability game $\mathcal{G}_{\vartheta}$.
  Then $\vartheta$ is satisfiable.
\end{theorem}
\begin{proof}
We treat the winning strategy, say $\sigma$, as a tree $T$ with nodes $\mathcal V$ and a root $r$.
The nodes are labelled with configurations corresponding to the strategy.
Thus, labels which belong to player 0 have at most one successor.
Only a node which is the objective of the rule \XruleWithE{} can
have more successors.

Let $\mathcal S$ be those nodes which are leaves or on which the rules \XruleNoE{} or \XruleWithE{}
are applied.  The tree defines a transition system as described just before of
\refsubsection{sub:winning:cond}.
Formally, for any node $s$ let $\widehat s$ be 
the oldest descendants ---including~$s$--- of~$s$ in~$\mathcal S$. 
Since player~0 owns all configurations besides those which rule \XruleWithE{} can handle,
Lemma~\ref{modal rule infinitely often} ensures that this assignment is total.
The edge relation $\mathord{\to} \subseteq \mathcal S \times \mathcal S$ is defined as 
\[
  \{(t,\widehat s) \mid s \text{ is a child of } t \text{ in }T\} 
  \; \cup \; 
  \{(s,s) \mid s \text{ is a leaf in }T\}
  \,\text.
\]
The induced transition system is $\Transsys_\vartheta = (\mathcal S, \widehat r, \to, \ell)$ where
$\ell(s) = \mathcal C \cap \Prop$ for any $s \in S$ labelled with a configuration $\mathcal C$.
Note that $\Transsys_\vartheta$ is total.
In the following, we omit the transition system $\Transsys_\vartheta$ in front of the symbol~$\models$.
Moreover, we identify a node with its annotated configuration.

For the sake of a contradiction, assume that $\Transsys_\vartheta, \widehat r \not\models \expath\vartheta$.
We will show that the winning strategy $\sigma$ admits an infinite play which contains a bad trace.
For this purpose, we simultaneously construct a maximal play $\mathcal C_0,
\mathcal C_1, \ldots$ which conforms to $\sigma$, a maximal connected
sequence of blocks $Q_0\Gamma_0, Q_1\Gamma_1, \ldots$ in this play, and a partial sequence $\pi_i$ of paths in $\Transsys_\vartheta$ such that the following properties hold for all indices $i$ and for all formulas $\varphi$ and $\psi$.
\begin{enumerate}[label=($\Xi$\hbox{-}\arabic*), labelindent=\parindent, leftmargin=2.7em]
  \item \label{soundness:trace:E}
    If $Q_i = \expath$ then $\widehat{\mathcal C_i} \not\models \expath(\bigwedge\Gamma_i)$.
  \item \label{soundness:trace:EQ}
    If $Q_i = \expath$, the rule \EErule{} or \EArule{} is applied to $\mathcal C_i$ with
    $\expath\Gamma_i$ and $\varphi$ as principals, and $\widehat{\mathcal C_i} \not\models \varphi$,
    then $Q_i \Gamma_i = \varphi$.
  \item \label{soundness:trace:A}
    If $Q_i = \allpath$ then $\pi_i$ is defined, $\widehat{\mathcal C_i} = \pi_i(0)$,
    and  $\pi_i \not\models \bigvee\Gamma_i$.
  \item \label{soundness:trace:AX}
    If $Q_i = Q_{i+1} = \allpath$, and the rule \XruleNoE{} or \XruleWithE{}
    is applied to $\mathcal C_i$ 
    then $\pi_{i+1} = \pi_i^1$ holds.  %
  \item \label{soundness:trace:nonAX}
    If $Q_i = Q_{i+1} = \allpath$, and neither
    \XruleNoE{} nor \XruleWithE{} is applied to $\mathcal C_i$
    then $\pi_{i+1} = \pi_i$ holds.   %
  \item \label{soundness:trace:AR}
    If $Q_i\Gamma_i = \allpath(\varphi \release \psi, \Sigma)$,
    if the rule \ARrule{} is applied to $\mathcal C_i$ such that $\varphi \release \psi$
    and $Q_i\Gamma_i$ are principal,
    and if $Q_{i+1}\Gamma_{i+1} = \allpath(\varphi, \nxt(\varphi \release \psi), \Sigma)$
    then $\pi_i \models \psi$.
\end{enumerate}

\noindent The construction is straight forward. We detail the proof for some cases,
and thereto use formulas and notations as shown in Figure~\ref{fig:pretableaurules}.
As for the rule \EArule{}, if $\widehat{\mathcal C_i} \not\models \expath(\allpath \varphi \wedge \bigwedge\Pi)$
then $\widehat{\mathcal C_i} \not\models \allpath \varphi$ or 
$\widehat{\mathcal C_i} \not\models \expath(\bigwedge\Pi)$.
If the first case does not apply then the trace is continued with $\expath \Pi$.
Otherwise, $Q_{i+1}\Gamma_{i+1} = \allpath \varphi$ holds and $\pi_{i+1}$ is
an arbitrary path in $\Transsys_{\vartheta}$ which starts at $\widehat{\mathcal C_i}$ %
and which fulfills $\pi_{i+1} \not\models \varphi$.
As for the rule \ARrule{}, we have $\pi_i \not\models \varphi \release \psi \vee \bigvee\Sigma$.
Using that $\pi_i = \pi_{i+1}$ and an unrolling of the $\release$-operator, 
$\pi_{i+1} \not\models \psi$ or $\pi_{i+1} \not\models \varphi \vee \nxt(\varphi\release\psi)$ holds.
In the first case the trace is continued with $\allpath(\psi, \Sigma$), and with 
$\allpath(\varphi,\nxt(\varphi\release\psi),\Sigma)$ otherwise.
As for case of \XruleNoE{} and \XruleWithE{}, the constraints determine the successor uniquely.

Back to the main proof: if the play is finite the last configuration consists
of literals only.  On the other hand, the last block of the sequence $\Xi$
reaches this leaf.  Therefore, the play must be infinite.  In particular, the sequence $\Xi$ is
a trace, and by Lemma~\ref{every trace form} it is either an $\expath$- or an $\allpath$-trace.
\begin{description}
\item[Trace $\Xi$ is an $\expath$-trace]
Let $n$ be minimal such that $(Q_i \Gamma_i, Q_{i+1} \Gamma_{i+1})$ 
is not spawning for all $i \geq n$.  Therefore, all these quantifiers $Q_i$s are $\expath$,
and the set $\Gamma_n$ is a singleton.
By~$\pi$ we denote the subsequence of the play $(\mathcal C_i)_{i \geq n}$ 
which consists of nodes in $\mathcal S$ only. 
For a node $\mathcal C$ in the play, we write $\pi^{\mathcal C}$ to denote the
suffix of $\pi$ starting at $\widehat{\mathcal C}$.
The trace contains a thread $\xi_0, \xi_1, \ldots$ such that
\begin{enumerate}[label=($\xi$\hbox{-}\arabic*), labelindent=\parindent, leftmargin=*]
\item \label{item:soundness:thread:notmodels}
  $\pi^{\mathcal C_i} \not\models \xi_i$, and
\item \label{item:soundness:thread:ER} 
  if $\xi_i = \varphi \release \psi$, the rule \ERrule{} is applied to $\mathcal C_i$
  with $\expath \Gamma_i$ and $\xi_i$ as principals, and $\pi^{\mathcal C_i} \not\models \psi$,
  then $\xi_{i+1} = \psi$.
\end{enumerate}
for all $i \geq n$ and all formulas $\varphi$ and $\psi$.
Indeed, the thread can be constructed step by step.  Obviously, there is a sequence of 
connected formulas $\xi_0, \ldots \xi_n$ within
the trace because the set $\Gamma_n$ is a singleton.
The rules \EOrrule{}, \EAndrule{}, \EUrule{} and \ERrule{} clearly preserve 
the properties~\ref{item:soundness:thread:notmodels} and~\ref{item:soundness:thread:ER}.
As for the rule \ELitrule{}, the formula $\xi_i$ cannot be the principal literal because 
$\pi^{\mathcal C_i}$ is a countermodel of $\xi_i$ but the literal survives until 
the next application of the model rules which defines the first state of $\pi^{\mathcal C_i}$.
If the rule \EErule{} or \EArule{} is applied, the property~\ref{soundness:trace:EQ}
keeps $\xi_i$ from being the principal formula because the considered suffix is the trace is not spawning.

By Lemma~\ref{every thread form}, $\xi$ is either a $\until$- or an $\release$-thread.
In the first case, the thread~$\xi$ attests that the trace is bad although player 0 wins the play.
Otherwise, suppose that $\xi$ is an $\release$-thread.
By Lemma~\ref{every UR-thread form}, there are $m \geq n$
and formulas $\varphi$ and $\psi$ such that 
$\xi_m = \varphi \release \psi$, and
$\xi_i = \varphi \release \psi$ or $\xi_i = \nxt (\varphi \release \psi)$ for all $i \geq m$,
Along the play $(\mathcal C_i)_{i \geq m}$, between any two consecutive 
applications of the rules \XruleNoE{} or \XruleWithE{},
the rule \ERrule{} must have been applied with $\xi_i = \varphi \release \psi$ and $Q_i\Gamma_i$ as
principals for some $i\geq m$. 
The property~\ref{item:soundness:thread:ER} ensures that
$\pi^{\mathcal C_i} \models \psi$.
Since this is true for any such two consecutive applications, 
$\pi^{\mathcal C_i} \models \psi$ for all $i\geq m$.
Therefore, $\pi^{\mathcal C_m}$ models 
$\varphi \release \psi$.
But this situation contradicts the property~\ref{item:soundness:thread:notmodels}
for $i$ being one the infinity many positions on which the rule \ERrule{} is applied to $Q_i \Gamma_i$ and $\xi_i$.

\item[Trace $\Xi$ is an $\allpath$-trace]
It suffices to show that $\Xi$ is a bad trace. 
Suppose for the sake of a contradiction  
that $\Xi$ contains an $\release$-thread $(\xi_i)_{i \in \Nat}$.
Let $n \in \Nat$ and $\varphi, \psi \in \fl{\vartheta}$
such that $Q_i = \allpath$, $\xi_n = \varphi \release \psi$,
and $\xi_i = \varphi \release \psi$
or $\xi_i = \nxt (\varphi \release \psi)$ for all $i \geq n$,
cf.\ Lemma~\ref{every UR-thread form}.

Along the play $(\mathcal C_i)_{i \geq n}$, between any two consecutive 
applications of the rules \XruleNoE{} or \XruleWithE{},
the rule \ARrule{} must have been applied such that $\xi_i$ and $Q_i\Gamma_i$ are principal
for some $i\in \Nat$. In this situation, the formula~$\xi_i$ is $\varphi \release \psi$.
Because $\xi_{i+1}$ is either $\varphi \release \psi$ or $\nxt(\varphi \release \psi)$,
the following element, $Q_{i+1}\Gamma_{i+1}$, of the trace is 
$\allpath(\varphi, \nxt(\varphi \release \psi), \Sigma)$ for some
$\Sigma \subseteq \fl{\vartheta}$.
Hence, thanks to~\ref{soundness:trace:AR} we have $\pi_i \models \psi$.
Because the block quantifier remains $\allpath$, the properties~\ref{soundness:trace:AX}
and~\ref{soundness:trace:nonAX} show that $\pi_n^j \models \psi$ for all $j\in\Nat$.
Therefore, $\pi_n \models \varphi \release \psi$ holds.
As the formula $\varphi \release \psi$ is $\xi_n$, the path $\pi_n$ satisfies $\bigvee \Gamma_n$.
However, this situation contradicts the property~\ref{soundness:trace:A}.
Thus, the considered play contains $\Xi$ as a bad trace.
\qedhere
\end{description}
\end{proof}

\subsection{Completeness}
To show completeness, we need a witness for satisfiable $\expath$-formulas.
For this purpose, let $\Transsys = (\States,s^*,\to,\lambda)$ be a transition system, $s \in \States$ be a state and
$\psi$ be a formula such that $s \models \expath\psi$.  We may assume a well-ordering $\lhd_\Transsys$
on the set of paths in~$\Transsys$~\cite{Zermelo04}.
The \emph{minimal $s$-rooted path that satisfies $\psi$}
is denoted by $\xi_\Transsys(s,\psi)$ and fulfills the following properties: $\xi_\Transsys(s,\psi)(0) = s$,
$\xi_\Transsys(s,\psi) \models \psi$, and there is no path $\pi$ with $\pi \lhd_\Transsys \xi_\Transsys(s,\psi)$,
$\pi(0) = s$ and $\pi \models \psi$.

A \emph{$\Transsys$-labelled (winning) strategy}
is a (winning) strategy with every configuration being labelled with a state such that the root is labelled with $s^*$, and for
every $s$-labelled configuration and every $s'$-labelled successor configuration it holds that $s \to s'$ if the corresponding
rule application is $\XruleWithE$ or $\XruleNoE$ and $s=s'$ otherwise.

\begin{theorem}%
\label{thm:completeness}
Let $\vartheta \in \ctlstar$ be satisfiable. 
Then player 0 has a winning strategy for the satisfiability game $\mathcal{G}_{\vartheta}$.
\end{theorem}

\begin{proof}
Let $\vartheta$ be a formula, $\Transsys = (\States,s^*,\to,\lambda)$ be a transition system, and $s^* \in \mathcal{S}$
be a state s.t.\ $\Transsys, s^* \models \expath\vartheta$.  In the following we may omit the system $\Transsys$ 
in front of the symbol~$\models$.

We inductively construct an $\States$-labelled strategy for player 0 as follows. Starting with the labelled configuration
$s^* : \expath\vartheta$, we apply the rules in an arbitrary but eligible ordering systematically by preserving
$s \models \Phi$ for every state-labelled configuration $s: \Phi$ and by preferring small formulas.  In particular,
the strategy is defined the following properties.  %
\begin{enumerate}[label=(S-\arabic*), labelindent=\parindent, leftmargin=2.7em]
\item \label{item:completeness:exists:Astate}
      If the rule application to follow $\Phi$ is $\ALitrule$, $\AErule$ or $\AArule$, with $\allpath(\psi, \Sigma)$ being
      the principal block in $\Phi$ and $\psi$ being the principal (state) formula, and $s \models \psi$,
      then the successor configuration of $\Phi$ follows $\psi$ and discards the original $\allpath$-block.
\item \label{item:completeness:exists:EU}
      If the rule application to follow $\Phi$ is $\EUrule$, with $\expath(\varphi \until \psi, \Pi)$ being
      the principal block in $\Phi$ and $\varphi \until \psi$ being the principal formula, then
      the successor configuration of $\Phi$ follows $\psi$ iff $\xi_\mathcal{T}(s, (\varphi \until \psi) \wedge \bigwedge \Pi) \models \psi$. 
\item If the rule application to follow $\Phi$ is $\EOrrule$, with $\expath(\psi_1 \vee \psi_2, \Pi)$ being
      the principal block in $\Phi$ and $\psi_1 \vee \psi_2$ being the principal formula, and
      $\xi_\mathcal{T}(s,(\psi_1 \vee \psi_2) \wedge \bigwedge \Pi) \models \psi_i$ for some $i \in \{1,2\}$, then the successor
      configuration of $\Phi$ follows $\psi_i$.
\item If the rule application to follow $\Phi$ is \XruleNoE{} and its successor configuration is labelled
      with a state $s'$ such that $s \to s'$ and successor configuration $\Phi'$ then player~0 labels this successor with the state $s'$.
\item \label{item:completeness:exists:E1}
      If player~1 applies the rule \XruleWithE{} to a configuration $\Phi$ which is labelled with a state $s$
      and obtains successor configuration $\expath\Pi, \Phi'$ then player~0 labels this successor with the state $\xi_\Transsys(s,\nxt\Pi)(1)$.
\end{enumerate}
Such a strategy exists because the property $s: \Phi$ can be maintained.
Note that every finite play ends in a node labelled with consistent literals only.
Clearly, player~0 wins such a play.

For the sake of contradiction, assume that player~0 does not win if she follows the strategy. Hence, there is an infinite labelled play
$s_0 : \Phi_0$, $s_1 : \Phi_1$, $\ldots$ (with $s_0 = s^*$ and $\Phi_0 = \expath\vartheta$) containing a bad trace $B_0, B_1, \ldots$.
We define a \emph{lift operation} $\widehat i$ that selects the next modal rule application as follows.
\begin{displaymath}
\widehat i := \min \{j \geq i \mid \Phi_j \textrm{ is the bottom of an application of $\XruleWithE$ or $\XruleNoE$}\}
\end{displaymath}
Due to Lemma~\ref{modal rule infinitely often}, $\widehat i$ is well-defined for every $i$.
Additionally, we define the \emph{modal distance} 
\begin{displaymath}
  \delta(i,k) := |\{j \mid i \leq j < k \text{ and } j = \widehat j\}|
\end{displaymath}
as well that counts the number of modal rule application between $i$ and $k$. Every position $i$ induces a generic path $\pi_i$ by
\begin{displaymath}
  \pi_i \colon j \mapsto s_{\min\{k \mid k \geq i \text{ and } \delta(i,k) =j\}}
\end{displaymath}
and note that the path $\pi_i$ indeed is well-defined for every $i$.

By Lemma~\ref{every trace form}, the bad trace is either an $\allpath$- or an
$\expath$-trace that is eventually not spawning, i.e.\ there is a position $n$ such that $B_i \equiv
\expath\Pi_i$ or $B_i \equiv \allpath\Sigma_i$ for all $i \geq n$ with $(B_i, B_{i+1})$ being not spawning. Let
$n$ be the least of such kind.

Next, the bad trace gives rise to a $\until$-thread in it that is satisfied by
the transition system. For this purpose we construct a $\until$-thread $\phi_0, \phi_1, \ldots$ in $B_0, B_1, \ldots$ 
such that all $i \geq n$ satisfy the following properties.

\begin{enumerate}[label=($\phi$\hbox{-}\arabic*), leftmargin=4em]
\item \label{item:completeness:Phi:model} $\pi_i \models \phi_i$.
\item \label{item:completeness:Phi:subf}
  For all formulas $\varphi$ and $\psi$ we have:
  If $\phi_i = \varphi \until \psi$, $\phi_i \not= \phi_{i+1}$ and if $\pi_i \models \psi$, then $\phi_{i+1} = \psi$.
\end{enumerate}
The construction of the thread depends on the kind of the trace.

\begin{description}
\item[Trace $B_0,B_1,\ldots$ is an $\expath$-trace]
The paths $\pi_i$ and $\xi_\Transsys(s_i, \Pi_i)$ coincide for all $i \geq n$ for two reasons.
First, whenever a rules besides \XruleNoE{} and \XruleWithE{} justifies the move from the configuration 
$\Phi_i$ to $\Phi_{i+1}$ for $i \geq n$, then $\xi_\Transsys(s_i, \Pi_i)$ and $\xi_\Transsys(s_{i+1}, \Pi_{i+1})$
are equal.
Second, this $\expath$-trace overcomes the application of the rules~\XruleNoE{} and \XruleWithE{}.
Thus, the minimal paths $\xi_\Transsys$ define the labels $s_n, s_{n+1}, \ldots$ 
and, in this way, the paths $\pi$.

Since $n$ is the \emph{least} index s.t.\ $(B_i, B_{i+1})$ is not spawning for all $i \geq n$, the set $\Pi_n$ has
to be a singleton. Define $\phi_n$ to be the single formula in $\Pi_n$.
Because $s_n \models \expath\Pi_n$, the path $\xi_\Transsys(s_n, \Pi_n)$
satisfies $\phi_n$.

As the trace is assumed to be bad, it contains a $\until$-thread, say $\phi_0, \phi_1, \ldots$.
The construction of the strategy ensures that
$\xi_\mathcal{T}(s_i, \Pi_i) \models \phi_i$ for all $i \geq n$.  Hence, $\pi_i \models \phi_i$.
Additionally, the constraint~\ref{item:completeness:exists:EU} yields the property~\ref{item:completeness:Phi:subf}.

\item[Trace $B_0,B_1,\ldots$ is an $\allpath$-trace]
Since $n$ is the least index such that $(B_i, B_{i+1})$ is not spawning $i \geq n$, the set $\Sigma_n$ has
to be a singleton. Define $\phi_n$ to be the single formula in $\Sigma_n$.

For $i\geq n$ the formula $\phi_{i+1}$ bases on $\phi_i$ as follows.
If $\widehat i = i$, that is, one of the modal rules \XruleNoE{} and \XruleWithE{} is to be applied next, 
then set $\phi_{i+1} = \phi'$ where $\phi_i = \nxt(\phi')$ for some formula $\phi'$.
Otherwise, $\widehat i \neq i$ holds. If $B_i$ or $\phi_i$ is not principal in the
rule instance then set $\phi_{i+1} := \phi_i$.  Because $(B_i, B_{i+1})$ is not spawning, 
$\phi_{i+1}$ belongs to $B_{i+1}$.  Otherwise, $B_i$ and $\phi_i$ are principal.
The formula $\phi$ is neither a literal nor an $\expath$- nor an $\allpath$-formula, 
because otherwise the property~\ref{item:completeness:exists:Astate} together with $\pi_i \models \phi_i$ would entail 
the end of this sequence of blocks or would show that the connection $(B_i, B_{i+1})$ is spawning.
Thus, the applied rule is either \ARrule{}, \AUrule{}, \AAndrule{} or \AOrrule{}.
If $\phi_i = \psi_1 \release \psi_2$ let $\phi_{i+1}$ be one of the successors $\psi'$ of $\phi_i$ contained in $B_{i+1}$ with 
$\pi_i \models \psi'$ and note that there is at least one. If
$\phi_i = \varphi \until \psi$, then set $\phi_{i+1} := \psi$ iff $\pi_i \models \psi$
and, otherwise, set $\phi_{i+1}$ to the other successor, that is $\varphi$ or $\nxt(\varphi \until \psi)$, of $\phi_i$ in $B_{i+1}$. 
Finally, if $\phi_i = \psi_1 \wedge \psi_2$ or $\phi_i = \psi_1 \vee \psi_2$, then set $\phi_{i+1} := \psi_k$
s.t.\ $\psi_k$ is connected to $\phi_i$ in $B_{i+1}$ and $\pi_i \models \psi_k$.

Putting suitable formulas in front of the sequence $\phi_n, \phi_{n+1}, \ldots$ entails
a thread in the trace $B_0, B_1, \ldots, B_n, B_{n+1}, \ldots$.  By assumption the trace is bad and by Lemma~\ref{every thread form},
the thread is a $\until$-thread.

\end{description}

\noindent Since $\phi_0, \phi_1, \ldots$ is a $\until$-thread, there are formulas $\varphi_0$ and $\varphi_1$ such that
$\phi_i = \varphi_0 \until \varphi_1$ for infinitely many indices $i$.
The set
\[
   A := \{i > n \mid \phi_{i-1} = \nxt(\varphi_1 \until \varphi_2) \textrm{ and } \phi_i = \varphi_1 \until \varphi_2 \}
\]
is infinite by Lemma~\ref{modal rule infinitely often} and~\ref{every UR-thread form}.
Let $i_0<i_1<\ldots$ be the ascending enumeration of $A$.
Between every two immediately consecutive elements either the rule $\XruleWithE$ or $\XruleNoE$  
is applied exactly once.   %
Therefore, $\pi_{i_j}^1 = \pi_{i_{j+1}}$ for all indices $j \geq 0$.
By property~\ref{item:completeness:Phi:model} we have 
$\pi_{i_0} \models \varphi_1 \until \varphi_2$.  Hence, there is a $k \geq 0$ such that 
$\pi_{i_0}^k \models \varphi_2$. 
In particular, $\pi_{i_k} \models \varphi_2$ and so $\pi_{i_k} \models \varphi_1 \until \varphi_2$.
For some $\ell$ between $i_k$ and $i_{k+1}$ the formula $\phi_\ell$ must be turned from 
$\varphi_1 \until \varphi_2$ into $\nxt(\varphi_1 \until \varphi_2)$
to finally pass the application of \XruleNoE{} and \XruleWithE{} at position $i_{k+1}-1$.  
However, the property~\ref{item:completeness:Phi:subf} shows that $\phi_\ell$ is just $\varphi_2$.
\end{proof}

\section{A Decision Procedure for \ctlstar}
\label{sec:decproc}

\subsection{Using Deterministic Automata to Check the Winning Condition}
\label{subsec:win with det automata}

Plays can be represented as infinite words over a certain alphabet, and we will show that the
language of plays that are won by player 0 is $\omega$-regular, i.e.\ recognisable by a nondeterministic B\"uchi automaton
for instance. 

The goal is then to replace the global condition on plays of having only good traces by an 
annotation of the game configurations with automaton states and a global condition on these states.
For instance, if the resulting automaton was of B\"uchi type, then the game would become a B\"uchi 
game: in order to solve the satisfiability game it suffices to check whether player 0 has a winning
strategy in the game with the annotations in the sense that she can enforce plays which are 
accepted by the B\"uchi automaton for the annotations. 

Now note that the automaton recognising such plays needs to be deterministic: suppose there
are two plays $uv$ and $uw$ with a common prefix $u$ s.t.\ both are accepted 
by an automaton $\mathcal{A}$. If $\mathcal{A}$ is nondeterministic then it may have two different
accepting runs on $uv$ and $uw$ that differ on the common prefix $u$ already. This could be resolved
by allowing two annotations on the nodes of the common prefix, but an infinite tree can have infinitely
many branches and it is not clear how to bound the number of needed annotations. However, if
$\mathcal{A}$ is deterministic then it has a unique run on the common prefix, and an annotation with
a single state of a deterministic automaton suffices.

It is known that every $\omega$-regular language can be recognised by a deterministic 
Muller \cite{IC::McNaughton1966}, Rabin~\cite{FOCS88*319} or parity automaton~\cite{conf/lics/Piterman06}.
A simple consequence of the last result is the fact that every game with an $\omega$-regular
winning condition can be reduced to a parity game. Thus, we could simply show that the winning
conditions of the satisfiability games of \refsection{sec:tableaux} are $\omega$-regular and
appeal to this result as well as known algorithms for solving parity games in order to have a
decision procedure for \ctlstar. While this does not seem avoidable entirely, it turns out that
the application of this technique, which is not very efficient in practice, can be reduced to a
minimum. The rest of this section is devoted to the analysis of the satisfiability games' winning 
conditions as a formal and $\omega$-regular language with a particular focus on the question of
determinisability. 

In our proposed reduction to parity games we will use annotations with states from two different 
deterministic automata: one checks that all $\expath$-traces are good, the other one checks that all 
$\allpath$-traces are good. The reason for this division is the fact that the former check is much 
simpler than the latter. It is possible to directly define a deterministic automaton that checks for 
absence of a bad $\expath$-trace. It is not clear at all though, how to directly define a deterministic 
automaton that checks for absence of a bad $\allpath$-trace. We therefore use nondeterministic automata 
and known constructions for complementing and determinising them. The next part recalls the automata theory
that is necessary for this, and in particular shows how these two automata used in the annotations
can be merged into one. 

\subsection{B\"uchi, co-B\"uchi and Parity Automata on Infinite Words}

We will particularly need the models of B\"uchi, co-B\"uchi and parity automata~\cite{lncs2500}. 

\begin{definition}
A \emph{nondeterministic parity automaton} (NPA) is a tuple
$\mathcal{A} = (Q,\Sigma,q_0,\delta,\Omega)$ with $Q$ being a finite set of \emph{states}, $\Sigma$ a 
finite \emph{alphabet}, $q_0 \in Q$ an \emph{initial state}, $\delta \subseteq Q \times \Sigma \times Q$
the \emph{transition relation} and $\Omega: Q \to \mathbb{N}$ a priority function. A \emph{run} of 
$\mathcal{A}$ on a $a_0 a_1 \ldots \in \Sigma^\omega$ is an infinite sequence $q_0,q_1,\ldots$ s.t.\ 
$(q_i,a_i,q_{i+1}) \in \delta$ for all $i \in \mathbb{N}$. It is accepting if $\max \{ \Omega(q) \mid q = q_i$ 
for infinitely many $i \in \mathbb{N} \}$ is even, i.e.\ if the maximal priority of a state that is seen
infinitely often in this run is even. The \emph{language} of the NPA $\mathcal{A}$ is $L(\mathcal{A}) = \{ w \mid$
there is an accepting run of $\mathcal{A}$ on $w \}$.
The \emph{index} of an NPA $\mathcal{A}$ is the number of different priorities occurring in it, i.e.\ 
$|\Omega[Q]|$. The \emph{size} of $\mathcal{A}$, written as $|\mathcal{A}|$, is the number of its states.

\emph{Nondeterministic B\"uchi} and \emph{co-B\"uchi automata} (NBA / NcoBA) are special cases of NPA.
An NBA is an NPA as above with $\Omega: Q \to \{1,2\}$, and an NcoBA is an NPA with $\Omega: Q \to \{0,1\}$.
Hence, an accepting run of an NBA has infinitely many occurrences of a state with priority $2$, and an
accepting run of an NcoBA has almost only occurrences of states with priority $0$. Traditionally, in an NBA 
the states with priority $2$ are called the \emph{final set}, and one defines an NBA as 
$(Q,\Sigma,q_0,\delta,F)$ where, in our terminology, $F := \{ q \in Q \mid \Omega(q) = 2 \}$. An NcoBA can equally
defined with an acceptance set $F$ rather than a priority function $\Omega$, but then $F := \{q \in Q \mid \Omega(q) = 0 \}$. 

An NPA / NBA / NcoBA with transition relation $\delta$ is \emph{deterministic} (DPA / DBA / DcoBA) 
if $|\{ q' \mid (q,a,q') \in \delta \}| = 1$ for all $q \in Q$ and $a \in \Sigma$.
In this case we may view $\delta$ as function from $Q \times \Sigma$ into $Q$.
\end{definition}

Determinism and the duality between B\"uchi and co-B\"uchi condition as well as the self-duality
of the parity acceptance condition makes it easy to complement a DcoBA to a DBA as well as a DPA 
to a DPA again. The following is a standard and straight-forward result~\cite[Section~1.2]{lncs2500} in the theory of
$\omega$-word automata.

\begin{lemma}
\label{lem:dpacomplement}
For every DcoBA, resp.\ DPA, $\mathcal{A}$ there is a DBA, resp.\ DPA, $\overline{\mathcal{A}}$ with 
$L(\overline{\mathcal{A}}) = \overline{L(\mathcal{A})}$ and $|\overline{\mathcal{A}}| = |\mathcal{A}|$.
\end{lemma}

In order to be able to turn presence of a bad trace---which may be easy to recognise using a nondeterministic
automaton---into absence of such which is required by the winning condition, we need complementation of 
nondeterministic automata as well. Luckily, an NcoBA can be determinised into a DcoBA using the 
Miyano-Hayashi construction~\cite{Miyano:1984:AFA} which can easily be complemented into a DBA
according to Lemma~\ref{lem:dpacomplement}.

\begin{theorem}[\cite{Miyano:1984:AFA}]
\label{thm:ncoba2dba}
For every NcoBA $\mathcal{A}$ with $n$ states there is a DBA $\overline{\mathcal{A}}$ with at most $3^n$ states 
s.t.\ $L(\overline{\mathcal{A}}) = \overline{L(\mathcal{A})}$.
\end{theorem}

NBA cannot be determinised into DBA, but into automata with stronger acceptance conditions 
\cite{FOCS88*319,conf/lics/Piterman06,conf/icalp/KahlerW08,conf/fossacs/Schewe09}. We are particularly 
interested in constructions that yield parity automata.

\begin{theorem}[\cite{conf/lics/Piterman06}]
\label{thm:nba2dpa}
For every NBA with $n$ states there is an equivalent DPA with at most $n^{2n+2}$ states and index at most 
$2n-1$.
\end{theorem}

For the decision procedure presented below we also need a construction that intersects a deterministic
B\"uchi and a deterministic parity automaton. This will allow us to consider absence of bad $\expath$- and bad
$\allpath$-traces separately.

\begin{lemma}
\label{lem:bucparintersect}
For every DBA $\mathcal{A}$ with $n$ states and DPA $\mathcal{B}$ with $m$ states and index $k$ there is a 
DPA $\mathcal{C}$ with at most $n\cdot m \cdot k$ many states and index at most $k+1$ s.t.\ 
$L(\mathcal{C}) = L(\mathcal{A}) \cap L(\mathcal{B})$.
\end{lemma}

\begin{proof}
Let $\mathcal{A} = (Q_1,\Sigma,q^0_1,\delta_1,F)$ and $\mathcal{B} = (Q_2,\Sigma,q^0_2,\delta_2,\Omega)$.
Define $\mathcal{C}$ as
$(Q_1 \times Q_2 \times \Omega[Q_2], \Sigma, (q^0_1,q^0_2,\Omega(q^0_2)), \delta, \Omega')$
where
\begin{displaymath}
\delta\big((q_1,q_2,p),a\big) \enspace := \enspace 
\begin{cases}
\big(\delta_1(q_1,a),\delta_2(q_2,a),\Omega(\delta_2(q_2,a))) &, \text{ if } q_1 \in F \\
\big(\delta_1(q_1,a),\delta_2(q_2,a),\max \{p, \Omega(\delta_2(q_2,a)) \}\big) &, \text{ if } q_1 \not\in F
\end{cases}
\end{displaymath}
Note that $\mathcal{C}$ simulates two runs of $\mathcal{A}$ and $\mathcal{B}$ in parallel on a word $w \in \Sigma^\omega$,
and additionally records in its third component the maximal priority that has been seen in $\mathcal{B}$'s run since the
last visit of a final state in the run of $\mathcal{A}$ if it exists. Thus, in order to determine whether or not both simulated runs
are accepting it suffices to examine the priorities at those positions at which the $\mathcal{A}$-component is visiting a final
state. In all other cases we choose a low odd priority. 
\begin{displaymath}
\Omega'(q_1,q_2,p) \enspace := \enspace 
\begin{cases}
p+2 &, \text{ if } q_1 \in F \\
1 &, \text{ if } q_1 \not\in F
\end{cases}
\end{displaymath}
Then the highest priority occurring infinitely often in a run of
$\mathcal{C}$ is even iff so is the one in the simulated run of $\mathcal{B}$ and $\mathcal{A}$ visits infinitely many final
states at the same time.

It should be clear that the number of states in $\mathcal{C}$ is bounded by $n\cdot m \cdot k$, and that it uses at most
one priority more than $\mathcal{B}$.
\end{proof}

To define an automaton which checks the absence of bad $\allpath$-traces, we need 
the intersection of B\"uchi with co-B\"uchi automata as well as alphabet projections of B\"uchi automata.

\begin{lemma}
\label{lem:DBAintersectDcoBA}
For every DBA $\mathcal A$ with $n$ states and every DcoBA $\mathcal B$ with $m$ states there is an NBA $\mathcal C$ 
with at most $n \cdot m \cdot 2$ states such that $L(\mathcal C) = L(\mathcal A) \cap L(\mathcal B)$.
\end{lemma}
\begin{proof}
  Let $\mathcal A$ be $(Q_1, \Sigma, q_1^0, \delta_1, F_1)$ and $\mathcal B$ be $(Q_2, \Sigma, q_2^0, \delta_2, F_2)$.
  Then define the NBA $\mathcal{C}$ as $(Q,\Sigma,(q_1^0, q_2^0,0),\delta,F_1 \times F_2 \times \{1\})$ with 
  $Q = (Q_1 \times Q_2 \times \{0\}) \cup (Q_1 \times F_2 \times \{1\})$, where $\delta$ realises the 
  synchronous product of $\delta_1$ and $\delta_2$ on $Q_1 \times Q_2 \times \{0\}$ and on $Q_1 \times F_2 \times \{1\}$.
  In addition, for every transition from $(q_1,q_2,0)$ to $(q'_1,q'_2,0)$ there is also one with the same alphabet symbol
  to $(q'_1,q'_2,1)$ if $q'_2 \in F_2$. Note that this creates nondeterminism.
\end{proof}

\begin{lemma}
\label{lem:NBAprojection}
Let $\mathcal C$ be an NBA over the alphabet $\Sigma_A \times \Sigma_B$.
There is a NBA $\mathcal A$ over the alphabet $\Sigma_A$ such that 
$|\mathcal A| \leq |\mathcal C|$ and
for all words $a_0 a_1 \ldots \in \Sigma_A^\omega$ it holds that
\[
  a_0 a_1 \ldots \in L(\mathcal A)
  \quad\text{ iff }\quad
  \text{ there is a word } b_0 b_1 \ldots \in \Sigma_B^\omega
  \text{ with } (a_0, b_0) (a_1, b_1) \ldots \in L(\mathcal C)
  \,\text.
\]
\end{lemma}
\begin{proof}
  The automaton $\mathcal C$ is almost $\mathcal A$.  Let $\delta_A$ be the transition relation of $\mathcal A$.
  Clearly, the set $\{(q,a,q') \mid (q,(a,b),q') \in \delta_A \text{ for some }b \in \Sigma_B\}$
  is adequate as a transition relation for $\mathcal C$.
\end{proof}

\subsection{An Alphabet of Rule Applications}
\label{subsec:decproc:alphabet}

Clearly, an infinite play in the game for some formula $\vartheta$ can be regarded as an $\omega$-word over
the alphabet of all possible configurations. This alphabet would have doubly
exponential size in the size of the input formula. It is possible to achieve the goals stated above with a
more concise alphabet. 

\begin{definition}
A \emph{rule application} in a play for $\vartheta$ is a pair of a configuration and one of its successors. 
Note that such a pair is entirely determined by the
principal block and the principal formula of the configuration and a number specifying the successor. This
enables a smaller symbolic encoding. For instance, the transition from the configuration
$\allpath(\expath\varphi,\Sigma),\Phi$ to the successor $\allpath\Sigma,\Phi$ in rule $\AErule$ can be
represented by the quadruple $(\allpath, \{\expath\varphi\}\cup \Sigma, \expath\varphi, 1)$. The other
possible successor would have index $0$ instead. There are three exceptions to this: applications of
rules $\Ettrule$ and $\XruleNoE$ can be represented using a constant name, and the successor in rule
$\XruleWithE$ is entirely determined by one of the $\expath$-blocks in the configuration. Hence, let
\[
\Sigma^{\mathsf{pl}}_\vartheta \enspace := \enspace 
  \big(\{ \allpath,\expath\} \times 2^{\fl{\vartheta}} \times \fl{\vartheta} \times \{0,1\}\big) \enspace 
  \cup \enspace \{ \expath\true, \nxt_0 \} \enspace \cup \enspace \big(\{\nxt_1\} \times 2^{\fl{\vartheta}}\big)
\]
Note that $|\Sigma^{\mathsf{pl}}_\vartheta| = 2^{\mathcal{O}(|\vartheta|)}$.
\end{definition}

An infinite play $\pi = C_0,C_1,\ldots$ then induces a word $\pi' = r_0,r_1,\ldots \in (\Sigma^{\mathsf{pl}}_\vartheta)^\omega$ 
in a straight-forward way: $r_i$ is the symbolic representation of the configuration/successor pair $(C_i,C_{i+1})$. We will not
formally distinguish between an infinite play $\pi$ and its induced $\omega$-word $\pi'$ over~$\Sigma^{\mathsf{pl}}_\vartheta$.

\phantomsection
\label{para:decproc:alphabet:con}  
For every $r \in \Sigma^{\mathsf{pl}}_\vartheta$ let $\mathrm{con}^{\expath}_r(\cdot)$ be a partial function
from $\expath$-blocks to $\expath$-blocks which satisfies the connection relation~$\leadsto$ and avoids spawning connections.
Thus, the function is undefined for $r=\expath\true$ and the argument $\expath\emptyset$, only.
For all other parameters and arguments the function is uniquely defined.

\subsection{DPA for the Absence of Bad \allpath-Traces} 
\label{subsec:decproc:Atraces}

An $\allpath$-trace-marked play is a (symbolic representation of a) play together with
an $\allpath$-trace therein.  It can be represented as an infinite word over the extended alphabet
\[
  \Sigma^{\mathsf{tmp}}_\vartheta := \Sigma^{\mathsf{pl}}_\vartheta \times \{\allpath, \expath\} \times 2^{\fl{\vartheta}}
  \,\text.
\]
The second and the last components of the alphabet simply name a block on the
marked trace.  Note that these components are half a step behind the first component because 
the latter links between two consecutive configuration.
Remember that an $\allpath$-trace can proceed through finitely many $\expath$-blocks before it
gets trapped in $\allpath$-blocks only.
We define a co-B\"uchi automaton $\mathcal{C}^{\allpath}_\vartheta$ which recognises exactly those 
$\allpath$-trace-marked plays which contain an $\release$-thread in the marked trace. It is 
$\mathcal{C}^{\allpath}_\vartheta = (\{\mathtt{W}\} \cup \flr{\vartheta},\Sigma^{\mathsf{tmp}}_\vartheta,\mathtt{W},\delta,\flr{\vartheta})$.
We describe the transition relation~$\delta$ intuitively. A formal definition can easily be derived from this. 
Starting in the waiting state~$\mathtt{W}$ it eventually guesses a formula of the form 
$\psi_1 \release \psi_2$ which occurs in the marked $\allpath$-trace. It then tracks this formula in
its state for as long as it is unfolded with rule \ARrule\ and remains in the marked trace. If it leaves
the marked trace in the sense that the trace proceeds through a block which does not contain this subformula
anymore, or an $\expath$-block occurs as part of the marked trace 
then $\mathcal{C}^{\allpath}_\vartheta$ simply stops. The following 
proposition is easily seen to be true using Definition~\ref{def:trace} and Lemma~\ref{every UR-thread form}.

\begin{lemma}
\label{lem:ncobaCcorrect}
Let $w \in (\Sigma^{\mathsf{tmp}}_\vartheta)^\omega$ be an $\allpath$-trace-marked play of a game for~%
$\vartheta$. Then $w \in L(\mathcal{C}^{\allpath}_\vartheta)$ iff the marked trace of $w$ contains an 
$\release$-thread.
\end{lemma} 

On the way to construct an automaton which recognises plays without bad $\allpath$-traces
we need to eliminate the restriction on $w$ in the previous lemma.
In other words, an automaton is needed which decides whether or not the annotated sequence of blocks
is an $\allpath$-trace.

\begin{lemma}
\label{lem:dcobaAtrace}
  There is a DcoBA $\mathcal{B}^\allpath_\vartheta$ over $\Sigma^{\mathsf{tmp}}_\vartheta$
  of size $\mathcal O(2^{|\vartheta|})$ such that the equivalence
  \[
    (( r_i, Q_i, \Delta_i ))_{i \in \Nat} \in L(\mathcal{B}^\allpath_\vartheta)
    \quad\text{ iff }\quad
    (Q_i \Delta_i)_{i \in \Nat} \text{ is an $\allpath$-trace in the play $(r_i)_{i \in \Nat}$}
  \]
  holds for all infinite plays $r_0, r_1, \ldots \in \Sigma^{\mathsf{pl}}_\vartheta$
  and all sequences of blocks $(Q_i \Delta_i)_{i \in \Nat}$.
\end{lemma}

\begin{proof}
Take as $\mathcal{B}^\allpath_\vartheta$ the deterministic co-B\"uchi automaton with states
\[
  Q := \{\expath, \allpath\} \times 2^{\fl{\vartheta}}
  \,\text,
\]
initial state $(\expath, \{\vartheta\})$ and final states 
$\{\allpath\} \times 2^{\fl{\vartheta}}$.
The automaton verifies that the last two components of the input indeed form an $\allpath$-trace.  
For this purpose, the state bridges between two successive blocks in the input sequence.
Due to the co-B\"uchi acceptance condition, the input is accepted if the block quantifier eventually
remains $\allpath$.  However, these properties define an $\allpath$-trace.

Formally, given a state $(Q_0, \Delta_0)$ and a letter $(r, Q_1, \Delta_1)$, a move
into the state $(Q_2, \Delta_2)$ is only possible iff 
$Q_0 = Q_1$, $\Delta_0 = \Delta_1$, and the rule instance $r$ transfers the block 
$Q_1 \Delta_1$ into the block $Q_2\Delta_2$.
Note that the sequence of blocks might end if the rules \Ettrule{} and \XruleWithE{}
are applied.  In such a situation, the automaton gets stuck and rejects thereby.
\end{proof}

\begin{figure}
  \resizebox{\textwidth}{!}{
    \hspace*{-1em}
    \bgroup\small


\def\TEXT#1#2{#1\,#2}
\def\LABEL{\small\vphantom{pb}} 
\def\LABELL{\vphantom{pb}} 
\begin{tikzpicture}[
    node distance=2cm,
    >=stealth,
    NODE/.style={
      rectangle,
      rounded corners=2pt,
      fill=black!10
    }
  ]

\node[]                                     (3) {};
\node[NODE, anchor=base, left=.5em of 3]    (1) {\TEXT{DBA for}{``marked $\allpath$-trace is bad'' over $\Sigma^{\mathsf{tmp}}_\vartheta$}};
\node[NODE, above of=1]                     (0) {\TEXT{NcoBA $\mathcal{C}^\allpath_\vartheta$ for}{``marked $\allpath$-trace is good'' over $\Sigma^{\mathsf{tmp}}_\vartheta$}};
\node[NODE, anchor=base, right=.5em of 3]   (2) {\TEXT{DcoBA $\mathcal{B}^\allpath_\vartheta$ for}{``marking is an $\allpath$-trace'' over $\Sigma^{\mathsf{tmp}}_\vartheta$}};
\node[NODE, anchor=left, below=1.8cm of 3]  (4) {\TEXT{NBA for}{``marking witnesses a bad $\allpath$-trace'' over $\Sigma^{\mathsf{tmp}}_\vartheta$}};
\node[NODE, anchor=left, below=1.5cm of 4]  (5) {\TEXT{NBA for}{``some $\allpath$-trace is bad'' over $\Sigma^{\mathsf{pl}}_\vartheta$}};
\node[NODE, below=2cm of 5]                 (6) {\TEXT{DPA $\mathcal{A}^\allpath_\vartheta$ for}{``all $\allpath$-traces are good'' over $\Sigma^{\mathsf{pl}}_\vartheta$}};

\coordinate (I) at ($(4.north)+(0,.6cm)$) ;  

\draw[->,thick] (0.south) -- (1.north) node [left,  midway] {\LABEL complementation\;}
                                       node [right, midway] {\LABEL \;Theorem~\ref{thm:ncoba2dba}} ;

\draw[-,thick]  (I) -| (1.south)
                (I) -| (2.south) ;
\draw[->,thick] (I) -| (4.north) ;
\node[anchor=south] at (I) {intersection, Lemma~\ref{lem:DBAintersectDcoBA}} ;

\draw[->,thick] (4.south) -- (5.north) node [right,midway] {\LABEL \;Lemma~\ref{lem:NBAprojection}}
                                       node [left, midway] {\LABEL alphabet projection\;} ;
\draw[->,thick] (5.south) -- (6.north) node [right,midway] {\LABEL \;\parbox{10em}{Theorem~\ref{thm:nba2dpa} and \LABELL\\ Lemma~\ref{lem:dpacomplement}\LABELL}} 
                                       node [left, midway] {\LABEL \parbox{10em}{\flushright determinisation and \LABELL\\  complementation \LABELL}} ;

\end{tikzpicture}

\egroup
    \hspace*{-1em}
  }  
\caption{Construction of the DPA for Theorem~\ref{thm:dpa for no bad A-traces}.}
\label{fig:decproc:Atraces}
\end{figure}

Figure~\ref{fig:decproc:Atraces} explains how the previously defined automata 
$\mathcal{C}^{\allpath}_\vartheta$ and $\mathcal{B}^{\allpath}_\vartheta$ can then be 
transformed into a deterministic parity automaton, called $\mathcal{A}^{\allpath}_\vartheta$, 
that checks for presence of an $\release$-thread in all $\allpath$-traces of a given play.
It is obtained using complementation twice, intersection and the projection of the alphabet 
$\Sigma^{\mathsf{tmp}}_\vartheta$ to $\Sigma^{\mathsf{pl}}_\vartheta$.
The four automata shown at the top are defined over the extended alphabet of plays with marked traces,
whereas the others work on the alphabet $\Sigma^{\mathsf{pl}}_\vartheta$ of symbolic rule applications only.
Almost all operations keep the automata small besides the determinisation.
All in all, we obtain the following property.

\begin{theorem}
\label{thm:dpa for no bad A-traces}
For every \ctlstar-formula $\vartheta$ there is a DPA $\mathcal{A}^{\allpath}_\vartheta$ 
of size $2^{2^{\mathcal O(|\vartheta|)}}$ and of index $2^{\mathcal O(|\vartheta|)}$
s.t.\ for all plays $\pi \in (\Sigma^{\mathsf{pl}}_\vartheta)^\omega$ we have: 
$\pi \in L(\mathcal{A}^\allpath_\vartheta)$ iff $\pi$ does not contain a bad $\allpath$-trace.
\end{theorem}

\subsection{DBA for the Absence of Bad \expath-Traces}
\label{subsec:decproc:Etraces}

Remember that a bad $\expath$-trace is one that contains a $\until$-thread. It is equally possible to construct 
an NcoBA which checks in a play for such a trace
and then use complementation and determinisation constructions as it is done for $\allpath$-traces. However, it is
also possible to define a DBA $\mathcal{A}_\vartheta^{\expath}$ directly which accepts a play iff it does not contain a bad $\expath$-trace. This
requires a bit more insight into the combinatorics of plays
but leads to smaller automata in the end.

Let $\varphi_0 \until \psi_0,\ldots, \varphi_{k-1} \until \psi_{k-1}$ be an enumeration of all $\until$-formulas in $\vartheta$.
The DBA $\mathcal{B}_\vartheta$ consists of the disjoint union of $k$ components $C_0,\ldots,C_{k-1}$ with 
$C_i = \{i\} \, \cup \, \{i\} \times 2^{\fl{\vartheta}}$.
In the $i$-th component, state $i$ is used to wait for either of two occurrences:
the $i$-th $\until$-formula gets unfolded or one of the rules for $\nxt$-formulas is being seen. In the first case the automaton 
starts to follow the thread of this particular $\until$-formula. In the second case, the automaton starts to look for the next
$\until$-formula in line to check whether it forms a thread. Hence, the transitions in state $i$ are the following. 
\begin{displaymath}
\delta(i,r) = \begin{cases}
(i,\Pi) & \text{ if } r = (\expath,\Pi,\varphi_i \until \psi_i,1) \\
(i+1)\bmod k & \text{ if } r = \nxt_0 \text{ or } r = (\nxt_1,\Pi) \text{ for some } \Pi \\
i & \text{ otherwise} 
\end{cases}
\end{displaymath}
In order to follow a thread of the $i$-th $\until$-formula, the automaton uses the states of the form $\{i\} \times 2^{\fl{\vartheta}}$
in which it can store the block that the current formula on the thread occurs in. It then only needs to compare this block to the
principal block of the next rule application to decide whether or not this block has been transformed. If it has been then the
automaton changes its state accordingly, otherwise it remains in the same state because the next rule application has left that
block unchanged. Once a rule application terminates the possible thread of the $i$-th $\until$-formula, the automaton starts 
observing the next $\until$-formula in line. There are two possibilities for this: either the next rule application fulfils the
$\until$-formula, or the $\expath$-trace simply ends, for instance through an application of rule $\XruleWithE$.
\[
  \delta((i,\Pi),r) = \begin{cases}    
    (i+1) \bmod k & \text{if }r=(\expath, \Pi, \varphi_i \until \psi_i, 0)
  \\
    (i, \Pi') & \text{otherwise, if }\mathrm{con}^{\expath}_r(\expath\Pi) = \expath\Pi'
  \\
    (i+1) \bmod k & \text{otherwise}
  \end{cases}
\]
where $\mathrm{con}^{\expath}_r$ is defined at the
end of \refsubsection{para:decproc:alphabet:con}.
The function $\delta$ is always defined as the second component of the state 
contains $\varphi_i \until \psi_i$ or $\nxt(\varphi_i \until \psi_i)$
whenever the first component is $i$. 

Note that there is no transition for the case of the next rule being $\XruleNoE$ because it only applies when there is no
$\expath$-block which is impossible if the automaton is following an $\until$-formula inside an $\expath$-trace.

It is helpful to depict the transition structure graphically.
\begin{center}
\begin{tikzpicture}[node distance=2cm]

  \node[state] (q0)               {$0$};
  \node[state] (q1) [right of=q0] {$1$};
  \node[state] (q2) [right of=q1] {$2$};
  \node[shape=circle,minimum size=8mm]        (i3) [right of=q2] {};
  \node        (ds) [right of=q2,node distance=2.5cm] {\ldots};
  \node[state] (qn) [right of=ds,node distance=2.5cm] {$k\!\!-\!\!1$};
  \node[shape=circle,minimum size=8mm]        (i4) [left of=qn]  {};

  \node        (c0) [above of=q0] {$C_0$};
  \node        (c1) [above of=q1] {$C_1$};
  \node        (c2) [above of=q2] {$C_2$};
  \node        (cn) [above of=qn] {$C_{k-1}$};

  \draw[rounded corners=3mm, thick=4pt, densely dotted] (-.8,-.6) rectangle (.8,2.6);
  \draw[rounded corners=3mm, thick=4pt, densely dotted] (1.2,-.6) rectangle (2.8,2.6);
  \draw[rounded corners=3mm, thick=4pt, densely dotted] (3.2,-.6) rectangle (4.8,2.6);
  \draw[rounded corners=3mm, thick=4pt, densely dotted] (8.2,-.6) rectangle (9.8,2.6);

  \path[->] (q0) edge (q1)
                 edge [loop below] ()
            (q1) edge (q2)
                 edge [loop below] ()
            (q2) edge (i3)
                 edge [loop below] ()
            (i4) edge (qn)
            (qn) edge [loop below] ();

  \path[->,draw] (qn) -- (8,-.6) .. controls (7,-1.2) .. (4.5,-1.2) .. controls (2,-1.2) .. (1,-.6) -- (q0);

  \path[->] (q0) edge (-.4,.9)
                 edge (0,.9)
                 edge (.4,.9)
            (q1) edge (1.6,.9)
                 edge (2,.9)
                 edge (2.4,.9)
            (q2) edge (3.6,.9)
                 edge (4,.9)
                 edge (4.4,.9)
            (qn) edge (8.6,.9)
                 edge (9,.9)
                 edge (9.4,.9);

  \path[->] (.6,.7) edge (q1)
            (.6,1.3) edge (q1)
            (.6,1.9) edge (q1)
            (2.6,.7) edge (q2)
            (2.6,1.3) edge (q2)
            (2.6,1.9) edge (q2)
            (4.6,.7) edge (i3)
            (4.6,1.3) edge (i3)
            (4.6,1.9) edge (i3)
            (7.6,.7) edge (qn)
            (7.6,1.3) edge (qn)
            (7.6,1.9) edge (qn);
\end{tikzpicture}
\end{center}
Note that every occurrence of rule $\XruleNoE$ or $\XruleWithE$ sends this automaton from any state~$i$ into the next
component modulo $k$. Furthermore, when unfolding the $i$-th $\until$-formula in state~$i$, it moves up into the
component $C_i$ where it follows the $\expath$-trace that it is in. From this component it can only get to state
$i+1 \bmod k$ if this $\until$-formula gets fulfilled. 

Thus, since any infinite play must contain infinitely many applications of rule $\XruleNoE$ or $\XruleWithE$, there
are only two possible types of runs of this automaton on such plays: those that eventually get trapped in some component
$C_i \setminus \{i\}$, and those that visit all of $0,1,\ldots,k-1$ infinitely often in this order.
  
It remains to be seen that this automaton---equipped with a suitable acceptance con\-dition---recognises exactly those
plays that do not contain a bad $\expath$-trace.

\begin{theorem}
\label{thm:dba for no bad E-traces}
For every \ctlstar formula $\vartheta$ with $k$ $\until$-subformulas there is a DBA $\mathcal{A}_\vartheta^{\expath}$ of 
size at most $k\cdot(1 + 2^{|\fl\vartheta|})$ s.t.\ for all plays $\pi \in \Sigma_\vartheta^\omega$: 
$\pi \in L(\mathcal{A}_\vartheta^{\expath})$ iff $\pi$ does not contain a bad $\expath$-trace.
\end{theorem}

\begin{proof}
As above, suppose that $\varphi_0\until\psi_0,\ldots,\varphi_{k-1}\until\psi_{k-1}$ are all the $\until$-formulas occurring
in $\fl{\vartheta}$. Let $\mathcal{A}_\vartheta^\expath := (C_0 \cup \ldots \cup C_{k-1},\Sigma_\vartheta,0,\delta,\{0\})$
be a B\"uchi automaton whose state set is the (disjoint) union of the components defined above and whose transition relation
$\delta$ is also as defined above. It is easy to check that $\mathcal{A}_\vartheta^\expath$ is indeed deterministic and of
the size that is stated above. It remains to be seen that it is correct.

Let $\pi$ be play. First we prove completeness, i.e.\ suppose that 
$\pi \not\in L(\mathcal{A}_\vartheta^\expath)$. Observe that in states of the form $i$ it can always react to any input
symbol whereas in states of the form $(i,\Pi)$ it can react to all input symbols apart from $\XruleNoE$. However, such states
are only reachable from states of the former type by reading a symbol of the form $(\expath,\Pi,\varphi\until\psi,1)$ which
is only possible when there is an $\expath$-block to which this rule is being applied. Furthermore, the automaton only stays
in such states for as long as this block still contains this $\until$-formula, and $\expath$-blocks can only disappear with 
rule $\Ettrule$ when they become empty. Thus, $\mathcal{A}_\vartheta^\expath$ has a (necessarily unique) run on every
play, and $\pi$ can therefore only be rejected if this run does not contain infinitely many occurrences of state $0$.

Next we observe that $\mathcal{A}_\vartheta^\expath$ cannot get trapped in a state of the form $i$ because every infinite
play contains infinitely many applications of rule $\XruleNoE$ or $\XruleWithE$---cf.~Lemma~\ref{modal rule infinitely often}---which 
send it to state
$(i+1)\bmod k$. Thus, in order not to accept $\pi$ it would have to get trapped in some component of states of the form
$(i,\Pi)$ for a fixed $i$. However, it only gets there when the $i$-th $\until$-formula gets unfolded inside an $\expath$-block,
and it leaves this component as soon as this formula gets fulfilled. Thus, if it remains inside such a component forever,
there must be an $\until$-thread inside $\expath$-blocks, i.e.\ a bad $\expath$-trace.

For soundness suppose that $\pi$ contains a bad $\expath$-trace. We claim that $\mathcal{A}_\vartheta^\expath$ must get
trapped in some component $C_i \setminus \{i\}$. Since this does not contain any final states, it will not accept $\pi$.
Now note that at any moment in a play, all $\until$-formulas which are top-level in some $\expath$-block need 
to be unfolded with rule $\EUrule$ before rule $\XruleNoE$ or $\XruleWithE$ can be applied. Thus, if $\mathcal{A}_\vartheta^\expath$
is in some state $i$, and the $i$-th $\until$-formula occurs inside an $\expath$-block at top-level position, then it will
move to the component $C_i\setminus\{i\}$ instead of to $(i+1)\bmod k$ because the latter is only possible with a rule that
occurs later than the rule which triggers the former transition.

As observed above, $\mathcal{A}_\vartheta^\expath$ cannot remain in the only final state $0$ forever. In order to visit it
infinitely often, it has to visit all states $0,1,\ldots,k-1$ infinitely often in this order. Thus, if there is a bad
$\expath$-trace with an $\until$-thread formed by the $i$-th $\until$-formula then there will eventually be a moment in
which this $i$-th $\until$-formula gets unfolded and $\mathcal{A}_\vartheta^\expath$ is trapped in some component
$C_j \setminus\{j\}$ for $j\ne i$ and the rest of the run, or it is in state $i$. If the latter is the case then it gets
trapped in $C_i \setminus \{i\}$ for the rest of the run before the next application of rule $\XruleNoE$ or $\XruleWithE$.
In either case, $\pi$ is not accepted.
\end{proof}

\subsection{The Reduction to Parity Games}
\label{subsec:decproc:games}

A \emph{parity game} is a game $\mathcal{G} = (V,V_0,E,v_0,\Omega)$ s.t.\ $(V,E)$ is a finite, directed graph with 
total edge relation $E$. $V_0$ denotes the set of nodes owned by player 0, and we write $V_1 := V \setminus V_0$
for its complement. The node $v_0 \in V$ is a designated starting node, and $\Omega: V \to \mathbb{N}$ assigns priorities to the nodes.
A \emph{play} is an infinite sequence $v_0,v_1,\ldots$ starting in $v_0$ s.t.\ $(v_i,v_{i+1}) \in E$ for all
$i \in \mathbb{N}$. It is won by player~$0$ if $\max \{ \Omega(v) \mid v=v_i$ for infinitely many $i \}$ is
even. A \emph{(non-positional) strategy} for player $i$ is a function $\sigma: V^*V_i \to V$, s.t.\ for all sequences
$v_0 \ldots v_n$ with $(v_j,v_{j+1}) \in E$ for all $j=0,\ldots,n-1$, and all $v_n \in V_i$ we have:
$(v_n,\sigma(v_0\ldots v_n)) \in E$. A play $v_0 v_1 \ldots$ \emph{conforms} to a strategy $\sigma$ for player $i$
if for all $j \in \mathbb{N}$ we have: if $v_j \in V_i$ then $v_{j+1} = \sigma(v_0\ldots v_j)$. A strategy $\sigma$ for player $i$
is a \emph{winning strategy} in node $v$ if player $i$ wins every play that begins in $v$ and conforms to $\sigma$.
A \emph{(positional) strategy} for player $i$ is a strategy $\sigma$ for player $i$ s.t.\ for all $v_0\ldots v_n \in V^*V_i$
and all $w_0\ldots w_m \in V^*V_i$ we have: if $v_n = w_m$ then $\sigma(v_0\ldots v_n) = \sigma(w_0\ldots w_m)$. Hence,
we can identify positional strategies with $\sigma: V_i \rightarrow V$.
It is a well-known fact that for every node $v \in V$, there is a winning strategy for either player~0 or player~1 for node $v$.
In fact, parity games enjoy positional determinacy meaning that there is even a positional winning strategy for node $v$
for one of the two player~\cite{focs91*368}.
The problem of \emph{solving} a parity game is to determine which player has a winning strategy for $v_0$.
It is solvable~\cite{Schewe/07/Parity} in time polynomial in~$|V|$ and exponential in~$|\Omega[V]|$.

\begin{definition}\label{def::sat_game}
Let $\vartheta$ be a state formula, $\mathcal{A}_\vartheta^\allpath$ be the DPA deciding
absence of bad $\allpath$-traces according to Theorem~\ref{thm:dpa for no bad A-traces},
$\mathcal{A}_\vartheta^\expath$ be the DBA deciding absence of bad $\expath$-traces
according to Theorem~\ref{thm:dba for no bad E-traces} and $\mathcal{A}_\vartheta = (Q, \Sigma^{\mathsf{pl}}_\varphi, q_0, \delta, \Omega)$ the
DPA recognising the intersection of the languages of $\mathcal{A}_\vartheta^\allpath$ and
$\mathcal{A}_\vartheta^\expath$ according to Lemma~\ref{lem:bucparintersect}.
The \emph{satisfiability parity game} for $\vartheta$ is  
$\mathcal{P}_\vartheta = (V,V_0,v_0,E,\Omega')$, defined as follows.
\begin{itemize}[leftmargin=1em]
\item $V  := \goals{\vartheta} \times Q$
\item $V_1 := \{(C,q) \in V \mid$ rule \XruleWithE{} applies to $C \}$
\item $V_0 := V \setminus V_1$
\item $v_0 := (\expath\vartheta,q_0)$
\item $((C,q),(C',q')) \in E$ iff $(C,C')$ is an instance of a rule application which is symbolically 
      represented by $r \in \Sigma^{\mathsf{pl}}_\vartheta$ and $q' = \delta(q,r)$,
	  or no rule is applicable to $C$ and $C = C'$ and $q = q'$, 
\item $\Omega'(C,q) := \begin{cases}
         0 & \textrm{if $C$ is a consistent set of literals} \\
         \Omega(q) & \textrm{if there is a rule applicable to $C$} \\
         1 & \textrm{otherwise}
       \end{cases}$
\end{itemize}
\end{definition}

The following theorem states correctness of this construction. It is not difficult to prove. In fact, 
winning strategies in the satisfiability games and the satisfiability parity games basically coincide. 

\begin{theorem}
\label{thm:gamescorrect}
Player 0 has a winning strategy for $\mathcal{P}_\vartheta$ iff player 0 has a winning strategy for $\mathcal{G}_\vartheta$.
\end{theorem}

\begin{proof}
Let $\pi$ be a play $(C_0,q_0), (C_1,q_1), \ldots$  of 
$\mathcal{P}_\vartheta$, and let $\pi' = C_0,C_1,\ldots$ be its projection onto the first components
which ends at the first configuration on which no rule can be applied.  The sequence $\pi'$ is indeed a 
play in $\mathcal{G}_\vartheta$. Note that this projection is invertible: for every play $\pi'$ in $\mathcal{G}_\vartheta$
there is a unique annotation with states of the deterministic automaton~$\mathcal{A}_\vartheta$ leading
to a play $\pi$ in $\mathcal{P}_\vartheta$. Now we have the following.
\begin{align*}
  \pi \text{\ is won by player } 0 
  \enspace &\Leftrightarrow \enspace 
  \pi' \text{\ is accepted by } \mathcal{A}_\vartheta \text{, or $\pi'$ ends in a consistent set of literals}
\\  
  \enspace &\Leftrightarrow \enspace 
  \pi' \text{\ is won by player } 0
\end{align*}
Thus, the projection of a winning strategy for player 0 in $\mathcal{P}_\vartheta$ is a winning strategy for her
in $\mathcal{G}_\vartheta$, and conversely, every winning strategy there can be annotated with automaton states in
order for form a winning strategy for her in $\mathcal{P}_\vartheta$.
\end{proof}

\begin{corollary}\label{cor:2EXPTIME}
Deciding satisfiability for some $\vartheta \in$ \ctlstar is in 2\EXPTIME. 
\end{corollary}

\begin{proof}
The number of states in $\mathcal{P}_\vartheta$ is bounded by
\begin{displaymath}
  |\goals{\vartheta}| 
  \enspace\cdot\enspace
  |Q|
  \enspace=\enspace
  2^{2^{\mathcal{O}(|\vartheta|)}} 
  \enspace\cdot\enspace
  2^{2^{\mathcal{O}(|\vartheta|)}}                        %
  \cdot 
  2^{\mathcal{O}(|\vartheta|)}                            %
  \cdot 
  |\vartheta| \cdot (1 + 2^{\mathcal{O}(|\vartheta|)})    %
  \enspace=\enspace 
  2^{2^{\mathcal{O}(|\vartheta|)}}
\end{displaymath}
Note that the out-degree of the parity game graph is at most $2^{|\vartheta|}$ because
of rule \XruleWithE. The game's index is $2^{\mathcal{O}(|\vartheta|)}$. 
It is known that parity games of size $m$ and index $k$ can be solved
in time $m^{\mathcal{O}(k)}$~\cite{Schewe/07/Parity} from which the
claim follows immediately.
\end{proof}

\subsection{Model Theory}
\label{subsec:decproc:modeltheory}

\begin{corollary}\label{cor:finite model}
Any satisfiable \ctlstar formula $\vartheta$ has a model of size at most $2^{2^{\mathcal{O}(|\vartheta|)}}$ and 
branching-width at most $2^{|\vartheta|}$.
\end{corollary}

\begin{proof}
Suppose $\vartheta$ is satisfiable. According to Theorems~\ref{thm:correctness} and~\ref{thm:gamescorrect}, 
player $0$ has a winning strategy for $\mathcal{P}_\vartheta$. It is well-known that she then also has a 
positional winning strategy~\cite{TCS::Zielonka1998}. A positional strategy can be represented as a finite
graph of size bounded by the size of the game graph. A model for $\vartheta$ can be obtained from this winning
strategy as it is done exemplarily in \refsection{sec:tableaux} and in detail in the proof of 
Theorem~\ref{thm:soundness}. The upper-bound on the branching-width is given by the fact that rule $\XruleWithE$ 
can have at most $2^{|\vartheta|}$ many successors. 
\end{proof}

The exponential branching-width stated in Corollary~\ref{cor:finite model} can be improved to a linear one by
restricting the rule applications.  The following argumentation
implicitly excludes the rules~\XruleNoE{} and~\XruleWithE{}.
Therefore, any considered rule application has exactly one principal
formula.

We limit the application of every rule besides~\XruleNoE{}
and~\XruleWithE{} to those applications where the principal formula is
a largest formula among those formulas in the configuration which do not have
$\nxt$ as their outermost connectives.
Following the proof of Theorem~\ref{thm:completeness}, any ordering
on the rules does not affect the completeness.

As a measure of a configuration we take the number of its $\expath$-blocks plus
the number of formulas having the form $\expath\varphi$ such that this
formula is a subformula, but not under the scope of an
$\nxt$-connective, of some formula in the configuration and such that
$\expath\{\varphi\}$ is not a block in this configuration.  This measure is
bounded by $|\vartheta|+1$ at the initial configuration $\expath\{\vartheta\}$ and at
every successor of the rules $\XruleNoE$ and $\XruleWithE$.

The size restriction ensures that any rule instance
apart from $\XruleNoE$ and $\XruleWithE$ weakly decreases the measure.
First, we consider the contribution of formulas to the measure.  An
inspection of the rules entails that any subformula $\expath\varphi$
which contributes to the measure of the configuration at the top of a rule occurs 
at the bottom as a subformula.  For the sake of contradiction, assume
that $\expath\varphi$ does not contribute to the measure of the configuration
at the bottom. Hence, the principal block is preventing $\expath\varphi$
from being counted and, hence, it has the shape
$\expath\{\varphi\}$. Therefore, the formula which hosts
$\expath\varphi$ is larger than the principal. But this situation
contradicts the size restriction.
Secondly, only the rules $\EErule$ and $\AErule$ can produce new $\expath$-blocks.  
If a formula $\expath\varphi$ is excluded from the measure of the configuration at the bottom
then and only then $\expath\{\varphi\}$ is a block in this configuration.
Therefore, in the positive case this block is not new at the top.
And in the negative case the new block at the top is paid by the
formula at the bottom and prevents other instances of this formula
at the top from being counted.

Putting this together with the argumentation in Corollary~\ref{cor:finite model}
yields the following.

\begin{corollary}
  \label{cor:finite model linear branch width}
  Any satisfiable \ctlstar formula $\vartheta$ has a model of size at most $2^{2^{\mathcal{O}(|\vartheta|)}}$ and 
  branching-width at most $|\vartheta|$.
\end{corollary}
These upper bounds are asymptotically optimal,
c.f.\,the proof of the 2\EXPTIME--lower-bound~\cite{STOC85*240} and
the satisfiable formula
$\bigwedge_{i=1}^n \expath\nxt (\neg p_i \wedge p_{i+1}) \;\wedge\; \bigwedge_{i=1}^n \allpath\nxt(p_i \to p_{i+1})$
which forces any model to be of branching-width~$n$.

\section{On Fragments of \ctlstar}
\label{sec:fragments}

The logic \ctlstar has two prominent fragments: \ctlplus and \ctl.
These logics allow refining the decision procedure detailed in
\refsection{sec:decproc}.
The obtained procedures are conceptionally simpler and have an optimal time-complexity.

\subsection{The Fragment \ctlplus} 
\label{sec:ctlplus}

The satisfiability problem for \ctlplus is 2\EXPTIME-hard~\cite{JL-icalp03}
and hence ---as a fragment of \ctlstar--- it is also 2\EXPTIME-complete.
Nevertheless, \ctlplus is as expressive as \ctl~\cite{EmersonHalpern85}.
Hence, the question arises whether the lower expressivity compared to \ctlstar
leads to a simpler decision procedure.

As \ctlplus is a fragment of \ctlstar we can apply the introduced games.
However, the occurring formulas will not necessarily be \ctlplus-formulas again,
because the fixpoint rules can prefix an  $\nxt$-constructor to the respective $\until$- or $\release$-formula.
Nevertheless, the grammar for \ctlplus can be expanded accordingly.
The new kinds are attached to line~\eqref{eq:ctlplus grammar:psi}.
\begin{align}\tag{\ref{eq:ctlplus grammar:psi}'}\label{eq:ctlplus grammar:psi'}
  \psi    \quad &{::=} \quad \varphi     \mid  \psi \lor \psi  \mid \psi \land \psi       \mid  \nxt \varphi     \mid  \varphi \until \varphi
  \mid  \varphi \release \varphi \mid \nxt (\varphi \release \varphi) \mid \nxt (\varphi \until \varphi)
\end{align}
The lines~\eqref{eq:ctlplus grammar:varphi} and~\eqref{eq:ctlplus grammar:psi'}
now define the grammar which every game follows.
The usage of these new formulas does not affect any of the used asymptotic measures.
The restriction to \ctlplus does not allow major simplification for the 
automata~$\mathcal{A}_\vartheta^{\expath}$ constructed in \refsubsection{subsec:decproc:Etraces}.  
However, the automata~$\mathcal{A}_\vartheta^{\allpath}$ 
which rejects plays containing bad $\allpath$-traces can be essentially simplified:
The refined construction bases on a coB\"uchi- instead of a B\"uchi-determinisation,
and hence leads to a simpler acceptance condition.
Due to Theorem~\ref{thm:ncoba2dba} it suffices to 
construct an exponentially sized NcoBA which detects an $\allpath$-trace which does
not contain any $\release$-thread.

For the rest of the subsection, fix a \ctlplus-formula $\vartheta$ and 
consider an infinite play in the game $\mathcal{G}_{\vartheta}$.
Let $(\qpath_i\Delta_i)_{i \in \Nat}$ be a trace in this play.
A position $i_0$ in this trace is called \emph{$\nxt$-stable} iff ---firstly--- the index $i_0$ addresses 
some top configuration either of the rule \XruleNoE{} or of \XruleWithE{}, and ---secondly---
the connection $\qpath_i\Delta_i \leadsto \qpath_{i+1}\Delta_{i+1}$ is not spawning
for every $i\geq i_0$.
By Lemma~\ref{modal rule infinitely often} and~\ref{every trace form} every trace 
has infinitely many $\nxt$-stable indices.

\begin{lemma}
  \label{lem:ctl+:seq state}
  Let $(\qpath_i\Delta_i)_{i \in \Nat}$ be a trace, let $i_0$ be one of its $\nxt$-stable positions,
  let $N \in \Nat$  and let $(\psi_i)_{i \leq N}$ be a sequence of connected formulas in the trace.
  If there is an $i_1 \geq i_0$ such that $\psi_{i_1}$ is a state formula then
  $\psi_{j}$ is a state formula for all~$j \geq i_1$.
\end{lemma}
\begin{proof}
  Every state formula in this trace eventually either
  disappears entirely ---by the rule~\ALitrule{} for instance---, 
  forms a new block \emph{outside} the trace ---by rule~\EErule{} for instance---,
  or get decomposed into a smaller state formula ---by rule~\EOrrule{} for instance---.
  One of these cases must happen before the rules~\XruleNoE{} or~\XruleWithE{} are applied.
  Finally, one of the two modal rules must be applied eventually due to Lemma~\ref{modal rule infinitely often}.
\end{proof}

For every thread Lemma~\ref{every UR-thread form} reveals 
a position which describes the corresponding suffix of the thread.
Next, we can strengthen this position to an  $\nxt$-stable position.

\begin{lemma}
  \label{lem:ctl+:thread}
  Let $(\qpath_i\Delta_i)_{i \in \Nat}$ be a trace and let $i_0$ be one of its $\nxt$-stable positions. 
  Every thread $(\psi_i)_{i \in \Nat}$ in the trace satisfies: 
  $\psi_i = \psi_{i_0}$ or $\psi_i = \nxt\psi_{i_0}$, for all $i \geq i_0$.
\end{lemma}
\begin{proof}
  The thread cannot hit any state formula, because by 
  Lemma~\ref{lem:ctl+:seq state} the thread would violate Lemma~\ref{every UR-thread form}.
  The application of the rule \XruleNoE{} or~\XruleWithE{} to the configuration at index $i_0-1$
  entails that $\psi_{i_0}$ is a $\until$- or an $\release$-formula.
  In particular along the remaining suffix, the thread must not hit a state formula.
  Therefore, the formula $\psi_i$ is either $\psi_{i_0}$ or $\nxt\psi_{i_0}$ for all $i \geq i_0$.
\end{proof}

\begin{theorem}
  \label{thm:ctl+:bad A trace}
  Let $(\qpath_i\Delta_i)_{i \in \Nat}$ be an $\allpath$-trace and let $i_0$ be one of its $\nxt$-stable positions. 
  We have that: the trace is bad,  iff $\Delta_i$ does not contain any $\release$- or $\nxt\release$-formula for some $i\geq i_0$.
\end{theorem}
\begin{proof}
  It suffices to show that the trace contains an $\release$-thread iff $\Delta_i$ contains a $\release$- or $\nxt\release$-formula 
  for every $i\geq i_0$. 
  The ``only if'' direction is a consequence of Lemma~\ref{lem:ctl+:thread}.
  As for the ``if'' direction, every $\release$- or $\nxt\release$-formula can be reached from the initial configuration of the game
  by a connected sequence of formulas.  Due to K\"onig's lemma there is a corresponding infinite sequence.
  By Lemma~\ref{every thread form}, this sequence is either a $\until$- or an $\release$-thread.
  If the latter case applies, we are done.  In the first case, infinitely many of the 
  said $\release$- and $\nxt\release$-formulas are reachable from a $\until$-formula.  
  Due to the grammar, a state formula must occur between the $\until$-formula
  and each of the considered $\release$- and $\nxt\release$-formulas.
  However, this situation contradicts Lemma~\ref{lem:ctl+:seq state}.
\end{proof}

The previous theorem is specific for \ctlplus.  For \ctlstar an $\allpath$-trace $(\qpath_i\Delta_i)_{i \in \Nat}$ can be good, 
even if $\Delta_i$ does not contain any $\release$- or $\nxt\release$-formula for some $i\geq i_0$.
Indeed, the $\release$-formula witnessing that the trace is good might be hosted within a $\until$-formula.
A play might delay the fulfillment of this $\until$-formula by several applications of \XruleNoE{} or~\XruleWithE{}.

Theorem~\ref{thm:ctl+:bad A trace} allows us to do without the determinisation
of B\"uchi-automata as used to construct $\mathcal{A}^{\allpath}_\vartheta$ in \refsubsection{subsec:decproc:Atraces}.
Indeed, there is a NcoBA which accepts every trace which contains a bad $\allpath$-traces.
Define the NcoBA $\mathcal{C}^{\allpath, \ctlplus}_\vartheta$ by $(Q,\Sigma^{\mathsf{br}}_\vartheta,\mathtt W,\delta,F)$ where
\[
  Q \enspace := \enspace \{\mathtt W\} \; \cup \; \Big(2^{\fl \vartheta} \times \{0,1,2\}\Big) \text,
\quad\text{ and }\quad
  F \enspace := \enspace 2^{\fl \vartheta} \times \{2\} \,\text.
\]
The automaton starts in the waiting state $\mathtt W$.
Every $\allpath$-trace contains a spawning connection for the last time ---
at least one such connection occurs because the initial configuration is an $\expath$-block.
This connection is generated either by the rule \AArule{} or by \EArule{}.
Thus, $\mathcal{C}^{\allpath, \ctlplus}_\vartheta$ eventually jumps after the corresponding input symbol, 
that is
$(\allpath, \placeholder, \allpath \varphi, 0)$
or
$(\expath, \placeholder, \allpath \varphi, \placeholder)$,
into the state $(\{\varphi\}, 0)$.
Then, $\mathcal{C}^{\allpath, \ctlplus}_\vartheta$ tries to successively 
guess an $\allpath$-trace using the first component.
If the block sequence stops or is spawning then the automaton rejects.
The value $0$ in the second component indicates the range between the last spawning
connection and the first application of rules~\XruleNoE{} and~\XruleWithE{} afterwards.
This application marks an $\nxt$-stable position.
The flags~$1$ and~$2$ are responsible for the remaining sequence starting with value~$1$.
The value is switched to~$2$ iff a block contains neither an $\release$- nor a $\nxt\release$-formula.
In such a situation, the automaton has to verify that the sequence does not break down.
Therefore, the final states of the NcoBA is defined as stated above.

The size of the automaton $\mathcal{C}^{\allpath, \ctlplus}_\vartheta$ is exponential in $|\vartheta|$.  
Hence, the complement of its Miyano-Hayashi determinisation is of double-exponential size 
---c.f.~Theorem~\ref{thm:ncoba2dba}--- 
and can be used in \refsubsection{subsec:decproc:games} instead of the general DPA~$\mathcal{A}^{\allpath}_\vartheta$.
Thus the time complexity of the whole decision procedure is double-exponential.

The advantage of this approach tailored to \ctlplus is the Miyano-Hayashi determinisation.
Their construction is simple to implement because it bases on an elaborated subset-construction only 
compared to known determinisation procedures for general B\"uchi automata~\cite{FOCS88*319, conf/lics/Piterman06}.

Because the small-formula strategy in Subsection~\ref{subsec:decproc:modeltheory} is indepenent of the fragement,
Corollary~\ref{cor:finite model linear branch width} also holds for \ctlplus.
The lower bound for the size is also doubly exponential~\cite{ipl-ctlplus08}.

\subsection{The Fragment \ctl}

The satisfiability problem for \ctl is \EXPTIME-complete.
Again, the question arises whether the lower expressivity compared to \ctlstar
leads to a simpler decision procedure.

As \ctl is a fragment of \ctlstar we could apply the introduced satisfiability
game. However,
this would lead to
games of doubly exponential size,
resulting in an unoptimal
decision procedure.

Hence, we define a new set of configurations and games rules that handle \ctl-formulas in an
optimal way. Due to the fact that subformulas of fixpoints in \ctl are always
state formulas, there is no need to keep the immediate subformulas in the respective
block after unfolding. By placing them at the top-level of the configurations, we can do
without the concept of blocks, since every block contains exactly one subformula.
Hence, these blocks can be understood as \ctl-formulas.

Here, a \emph{configuration (for $\vartheta$)} is a non-empty set of state
formulas of the set $\{\varphi, \expath\nxt \varphi, \allpath\nxt \varphi \mid \varphi \in \subf{\vartheta}\}$.
The additional formulas $\expath\nxt \varphi$ and $\allpath\nxt \varphi$ will be generated when unfolding fixpoints.  
In return, the Fischer-Ladner closure is replaced with the set of subformulas.
The definition of consistency etc.\ is exactly the same as before.

Again, we write $\goals{\vartheta}$ for the set of all consistent configurations for $\vartheta$. Note that this is a
finite set of at most exponential size in $|\vartheta|$.

\begin{figure}[t]%
\[
  \unaryRule{\Andrule}
    {\varphi_1, \varphi_2}
    {\varphi_1 \land \varphi_2}
\qquad
  \choiceRule{\Orrule}
	{\varphi_1}
    {\varphi_2}
    {\varphi_1 \lor \varphi_2}
\]
\[
  \choiceRule{\EUrule}
    {\varphi_2}
    {\varphi_1,\expath\nxt\expath(\varphi_1 \until \varphi_2)}
    {\expath(\varphi_1 \until \varphi_2)}
\qquad
  \choiceRule{\AUrule}
    {\varphi_2}
    {\varphi_1,\allpath\nxt\allpath(\varphi_1 \until \varphi_2)}
    {\allpath(\varphi_1 \until \varphi_2)}
\]
\[
  \choiceRule{\ERrule}
    {\varphi_1,\varphi_2}
    {\varphi_2,\expath\nxt\expath(\varphi_1 \release \varphi_2)}
    {\expath(\varphi_1 \release \varphi_2)}
\qquad
  \choiceRule{\AUrule}
    {\varphi_1,\varphi_2}
    {\varphi_2,\allpath\nxt\allpath(\varphi_1 \release \varphi_2)}
    {\allpath(\varphi_1 \release \varphi_2)}
\]
\[
  \def\addSideFormulas{}
  \unaryRule{\XruleNoE}
    {\varphi_1, \ldots, \varphi_n}
    {\allpath\nxt\varphi_1, \ldots, \allpath\nxt\varphi_n, \Lambda}
\qquad	
  \branchingRule{\XruleWithE}
    {\varphi_1', \varphi_1, \ldots, \varphi_n}
    {\varphi_m', \varphi_1, \ldots, \varphi_n}
    {\expath\nxt\varphi_1', \ldots, \expath\nxt\varphi_m', \allpath\nxt\varphi_1, \ldots, \allpath\nxt\varphi_n, \Lambda}
\]

\caption{The game rules for \ctl.}
\label{fig:pretableaurulesctl}
\end{figure}

\begin{definition}
The satisfiability game for a \ctl-formula $\vartheta$ is a directed graph $\mathcal{G}_{\vartheta} = (\goals{\vartheta},V_0,E,v_0,L)$
whose nodes are all possible configurations and whose edge relation is given by the game rules in Figure~\ref{fig:pretableaurulesctl}.
It is understood that the formulas which are stated explicitly under the line do not occur 
in the sets~$\Lambda$ or~$\sideFormulas$. The symbol $\ell$ stands for an arbitrary literal.
The initial configuration is $v_0 = \vartheta$. The winning condition $L$ will be described next.
\end{definition}

Again, we need to track the infinite behaviour of eventualities. However, the
situation is much easier here. First, we can do without the concept of blocks,
implying that we can do without the concept of traces as well. Second, there is
no structural difference in tracking bad threads contained in $\expath$- or
$\allpath$-blocks, any infinite trace contains exactly one thread, i.e.\
existential quantification and universal quantification over threads in traces
are interchangeable.

The definition of principal formulas and plays is the same as before, and we
again have the definition of connectedness and write
$(\mathcal{C},\varphi) \leadsto (\mathcal{C}',\varphi')$ to indicate that
$\varphi \in \mathcal{C}$ if connected to the subsequent formula $\varphi' \in \mathcal{C}'$.
There
are still infinitely many applications of rules $\XruleNoE$ or $\XruleWithE$ in
a play.

\begin{definition}
Let $\mathcal{C}_0, \mathcal{C}_1, \ldots$ be an infinite play.
A \emph{thread} $t$ within $\mathcal{C}_0, \mathcal{C}_1, \ldots$ again is an
infinite sequence of formulas $\varphi_0, \varphi_1, \ldots$ s.t.\ for all
$i \in \Nat$: $(\mathcal{C}_i,\varphi_i) \leadsto (\mathcal{C}_{i+1},\varphi_{i+1})$.

Again, such a thread $t$ is called a $\until$-thread, resp.\ an $\release$-thread if there is a formula
$\varphi \until \psi \in \subf{\vartheta}$, resp.\ $\varphi \release \psi \in \subf{\vartheta}$ s.t.\
$\psi_j = \varphi \until \psi$, resp.\ $\psi_j = \varphi \release \psi$ for infinitely many~$j$.
\end{definition}

Again, every play contains a thread and every thread is either an $\until$-thread
or a $\release$-thread.

\begin{definition}
An infinite play $\pi = C_0,C_1,\ldots$ belongs to the winning condition $L$ of $\mathcal{G}_{\vartheta} = (\goals{\vartheta},V_0,E,v_0,L)$
if $\pi$ does not contain a $\until$-thread.
\end{definition}

The following can be shown in similar way as Theorem~\ref{thm:correctness}:

\begin{theorem}
\label{thm: ctl correctness}
For all $\vartheta \in$ \ctl: $\vartheta$ is satisfiable iff player 0 has a winning strategy for the satisfiability game
$\mathcal{G}_{\vartheta}$.
\end{theorem}

As decision procedure, we again propose to apply a reduction to parity games,
similar to the one of \refsubsection{subsec:decproc:games}. The parity game is
constructed the same way by using the reduced configuration set of this section.
Additionally, we can construct a much simpler DPA for checking the winning
conditions.

Due to the fact that there are no traces anymore resp.\ every trace now contains
a thread-singleton, we can either apply an automaton construction similar to
the one of \refsubsection{subsec:decproc:Atraces} or to the one of
\refsubsection{subsec:decproc:Etraces}. We follow the latter approach here.

Remember that the automaton of \refsubsection{subsec:decproc:Etraces} was composed
of the disjoint union of $k$ components $C_0,\ldots,C_{k-1}$ with 
$C_i = \{i\} \,\cup\, \{i\} \times 2^{\fl{\vartheta}}$, where
$\varphi_0 \until \psi_0,\ldots, \varphi_{k-1} \until \psi_{k-1}$ was an
enumeration of all $\until$-formulas in $\vartheta$.
We can simplify the automaton dramatically here by considering the components
$C_i = \{i\} \,\cup\, \{i\} \times \{\varphi, \expath\nxt\varphi, \allpath\nxt\varphi \mid \varphi \in \subf{\vartheta}\}$ 
instead. 
The transition function is updated accordingly, following now single formulas instead of blocks. 
We can get a result similar to Theorem~\ref{thm:dba for no bad E-traces}:

\begin{theorem}
\label{thm:dba for no ctl threads}
For every \ctl formula $\vartheta$ with $k$ $\until$-subformulas there is a DBA $\mathcal{A}$ of 
size at most $k\cdot(1 + 3|\vartheta|)$ s.t.\ for all plays $\pi$: 
$\pi \in L(\mathcal{A})$ iff $\pi$ does not contain a $\until$-thread.
\end{theorem}

By attaching this automaton to our parity game, we obtain an optimal decision
procedure for \ctl:

\begin{corollary}\label{cor:EXPTIME ctl}
Deciding satisfiability for some $\varphi \in$ \ctl is in \EXPTIME. %
\end{corollary}

\begin{proof}
The number of states in the constructed parity game is bounded by
\begin{displaymath}
  2^{\mathcal{O}(|\vartheta|)} \cdot |\vartheta| \cdot (1 + 3|\vartheta|) 
  \enspace=\enspace
  2^{\mathcal{O}(|\vartheta|)}
\,\text.
\end{displaymath}
Note that the out-degree of the parity game graph is at most $|\vartheta|$ because
of rule \XruleWithE\, which is bounded by the number of $\expath$-formulas in 
$\vartheta$. The game's index is $2$ which makes it, in fact, a
B\"uchi game.
It is well-known \cite{ChatterjeeHenzingerPiterman06_AlgorithmsForBuchiGames}
that B\"uchi games with $n$ states and $m$ edges can be solved in time $\mathcal{O}(n \cdot m)$ from which the
claim follows immediately.
\end{proof}

The previous upper bound is optimal because 
the satisfiability problem for \ctl-fragment \pdl is \EXPTIME-hard~\cite{Fischer79}.
Since each block in the configurations is mainly a subformula of $\vartheta$,
the branching-width is bounded by $|\vartheta|$.  This bound is independent of
the strategy as compared with Corollary~\ref{cor:finite model linear branch width}.

\section{Comparison with Existing Methods}
\label{sec:compare}

\subsection{\ctlstar}

We compare the game-based approach with existing decision
procedures for \ctlstar, namely Emerson/Jutla's tree automata~\cite{Emerson:2000:CTA}, Kupferman/Vardi's 
automata reduction~\cite{conf/focs/KupfermanV05}, Reynolds' proof
system~\cite{Reynolds00}, %
and Reynolds' tableaux~\cite{conf/fm/Reynolds09} with respect to several aspects like computational
optimality, availability of an implementation etc., c.f.\ Table~\ref{fig:cmp}.

\afterpage{
\begin{table}[t]
{\small\noindent
\def\stacktext##1{{\def\arraystretch{1}\tabular{@{}c@{}}##1\endtabular}} %
\renewcommand{\arraystretch}{1.5}%
\begin{tabularx}{\linewidth}{>{\raggedright}X@{\qquad}c@{\qquad}c@{\qquad}c@{\qquad}c}
\toprule
\backslashbox[4.0cm]{Aspect}{Method} %
     & \stacktext{Emerson \\\& Jutla\\\cite{Emerson:2000:CTA}} 
     & \stacktext{Reynolds\\\cite{Reynolds00}} 
     & \stacktext{Kupferman \& Vardi \\(\& Wolper)\\ \cite{KVW00,conf/focs/KupfermanV05}}  & here
\\\midrule
Concept                       & automata                 & tableaux         & automata-reduction & games                    \\
Worst-case complexity         & 2\EXPTIME                & 2\EXPTIME        & 2\EXPTIME & 2\EXPTIME                \\
Implementation available      & no                       & yes              & no & yes\footnotemark{}                      \\
Model construction            & yes                      & yes              & no & yes                      \\
Out-degree                    & $\mathcal{O}(n)$         & $\mathcal{O}(n)$ & $\mathcal{O}(n)$ & $\mathcal{O}(n)$         \\
Requires small model property & no                       & yes              & no & no                       \\
Derives small model property  & $2^{2^{\mathcal{O}(n)}}$ & ---              & $2^{2^{\mathcal{O}(n)}}$ & $2^{2^{\mathcal{O}(n)}}$ \\
Needs B\"uchi determinisation & yes                      & no               & no & yes                      \\
\bottomrule
\end{tabularx}%
}
\caption{Comparison of the main decision methods for satisfiability in \ctlstar.}
\label{fig:cmp}
\end{table}
\footnotetext{\url{https://github.com/oliverfriedmann/mlsolver}}
}%

Emerson/Jutla's procedure transforms a \ctlstar-formula $\varphi$ in some normal form into a
tree-automaton recognising exactly the tree-unfoldings of fixed bran\-ching-width of all models of
$\varphi$. This uses a translation of linear-time formulas into B\"uchi automata and then into
deterministic (Rabin) automata for the same reasons as outlined in \refsubsection{subsec:win with det automata}. 
The game-based approach presented here does not use tree-automata as such, but player-0-strategies
resemble runs of a tree automaton. The crucial difference is the separation between the use of machinery
for the characterisation of satisfiability in \ctlstar and the use of automata only in order to make
the abstract winning conditions effectively decidable. In particular, we do not need translations of
linear-time temporal formula into $\omega$-word automata.  The relationship between input formula and
resulting structure (here: game) is given by the rules. Furthermore, this separation enables the 
branching-width of models of $\varphi$ to be flexible; it is given by the number of successors of the 
rule \XruleWithE{}. In a tree automaton setting it is a priori fixed to a number which is linear in the 
size of the input formula. While this does not increase the asymptotic worst-case complexity, it may have 
an effect on the efficiency in practice. Not surprisingly, we do not know of any attempt to implement the
tree-automata approach.

Kupferman/Vardi's approach is not just a particular decision procedure for \ctlstar. Instead, it is a general
approach to solving the emptiness problem for alternating parity tree automata. While this can generally
be done using determinisation of B\"uchi automata as in Emerson/Jutla's approach, Kupferman/Vardi have
found a way to avoid B\"uchi determinisation by using universal co-B\"uchi automata instead. These are
translated into alternating weak tree automata and, finally, into nondeterministic B\"uchi tree automata.
Emptiness of the latter is relatively easy to check. In the case of \ctlstar, a formula $\varphi$ can be
translated into a hesitant alternating automaton of size $\mathcal{O}(|\varphi|\cdot 2^{|\varphi|})$ 
\cite{KVW00} whose emptiness can be checked in time that is doubly exponential in $|\varphi|$.

The price to pay, though, is the use of a reduction that is only satisfiability-preserving. Thus, their 
approach reduces the satisfiability problem for branching-time temporal logics that can be translated into 
alternating parity tree automata to the emptiness problem for tree automata which accept some tree iff the input
formula is satisfiable. The translation does not preserve models, though. There is a way of turning
a tree model for the nondeterministic B\"uchi automaton back into a tree model for the branching-time temporal
logic formula because the alphabet that the universal co-B\"uchi automaton uses is just a projection of the 
hesitant alternating tree automaton's alphabet. Still, this procedure does not seem to keep a close connection 
between the subformulas of the input formulas and the structure of the resulting tree automaton which is being
checked for emptiness.

Reynolds' proof system \cite{Reynolds00} is an approach at giving a sound and complete finite axiomatisation for
\ctlstar. Its proof of correctness is rather intricate and the system itself is useless for practical
purposes since it lacks the subformula property and it is therefore not even clear how a decision procedure,
i.e.\ proof search could be done. In comparison, the game-based calculus has the subformula property---formulas in
blocks of successor configurations are subformulas of those in the blocks of the preceding one---and comes
with an implementable decision procedure. The only price to pay for this is the characterisation of 
satisfiability through infinite objects instead.

Reynold's tableau system \cite{Rey:startab} shares some similarities with the games presented here. He also uses sets
of sets of formulas as well as traces (which he calls threads), etc. Even though his tableaux
are finite, the difference in this respect is marginal. Finiteness is obtained through looping back,
i.e.\ those branches might be called infinite as well. One of the real differences between the
two systems lies in the way that the semantics of the \ctlstar operators shows up. In Reynolds' system
it translates into technical requirements on nodes in the tableaux, whereas the games come with
relatively straight-forward game rules. The other main difference is the loop-check. Reynolds says
that ``\ldots {\it we are only able to give some preliminary results on mechanisms for tackling repetition.}
[\ldots] {\it The task of making a quick and more generally usable repetition checker will be left to be 
advanced and presented at a later date.}'' The game-based method comes with a non-trivial repetition checker: it
is given by the annotated automata.

\subsection{The Fragments \ctlplus and \ctl}

\begin{table}[t]
{\small\noindent
\def\stacktext#1{{\def\arraystretch{1}\tabular{@{}c@{}}#1\endtabular}}
\renewcommand{\arraystretch}{1.5}%
\begin{tabularx}{\linewidth}{>{\raggedright}X@{\quad\enspace}c@{\quad\enspace}c@{\quad\enspace}c@{\quad\enspace}c}
\toprule
\backslashbox[4.0cm]{Aspect}{Method} %
                              & \stacktext{Emerson \& Halpern\\\cite{EmersonHalpern85}}         
                                                        & \stacktext{Vardi \& Wolper\\\cite{VW86a}}
                                                                                & \stacktext{Abate et al.\\\cite{conf/lpar/AbateGW07}}
                                                                                                             & here
\\\midrule
Concept                       & filtration              & automata              & tableaux                   & games \\
Worst-case complexity         & \EXPTIME                & \EXPTIME              & 2\EXPTIME                  & \EXPTIME \\
Implementation available      & no                      & no                    & yes                        & yes \\
Model construction            & yes                     & yes                   & yes                        & yes \\
Requires small model property & yes                     & no                    & no                         & no \\
Derives small model property  & ---                     & $2^{{\mathcal O}(n)}$ & $2^{2^{{\mathcal O}(n)}}$  & $2^{{\mathcal O}(n)}$ \\
\bottomrule
\end{tabularx}%
}
\caption{Comparison of the main decision methods for satisfiability of \ctl-formulas.}
\label{fig:cmp ctl}
\end{table}

To the best of our knowledge, there are no decision procedures that are especially tailored towards \ctlplus.
Thus, the restriction of the satisfiability games to \ctlplus as presented in \refsection{sec:ctlplus} is the
first decision procedure for this logic which does not also decide the whole of \ctlstar.

The situation for \ctl is entirely different. The first decision procedure for \ctl was given by Emerson and
Halpern~\cite{EmersonHalpern85} using filtration. It starts with a graph of Hintikka sets and successively 
removes edges from this graph in order to exclude unfulfilled eventualities. This is similar to the game-based
approach in that the game rules for Boolean connectives mimic the rules for being a Hintikka set. On the other
hand, the machinery for excluding unfulfilled eventualities is an entirely different one. 

There is a purely automata-theoretic decision procedure for \ctl~\cite{VW86a}: as such, it constructs a tree
automaton which recognises all tree-unfoldings of models of the input formula. In order to obtain an asymptotically
optimal decision procedure for \ctl, Vardi/Wolper use a new type of acceptance condition resulting in 
\emph{eventuality automata} whose emptiness problem can be decided in polynomial time. An exponential 
translation from \ctl into such automata then yields a decision procedure for \ctl. There are certain 
similarities to the game-based approach presented here: the design of the simpler type of acceptance condition
is reminiscent of the manual creation of deterministic automata that check the winning conditions.

There is a tableau-based decision procedure for \ctl~\cite{conf/lpar/AbateGW07}. As with Reynold's tableaux
for \ctlstar, the main difference to the game-based (and also automata-theoretic) approach is the fact that
the tableau calculi do not separate the decision procedure into a syntactical characterisation (e.g.\ winning 
strategy) and an algorithm deciding existence of such objects. This leads to correctness proofs which are even
more complicated than the ones for the \ctlstar games presented here. Also, this method does not yield a
common framework for dealing with unfulfilled eventualities which is given by the different types of 
(deterministic) automata which are being used here in order to characterise the winning conditions. 

The work that is most closely related to the one presented here consists of the focus game approach to \ctl~%
\cite{ls-lics2001}. These are also satisfiability games, and the rules there extend the rules here with a focus on 
a particular subformula which is under player 1's control. The focus game approach does not explicitly give an
algorithm for deciding satisfiability. A close analysis shows that the focus can be seen as an annotation with 
a nondeterministic co-B\"uchi automaton to the game configurations, and a decision procedure could be obtained by
determinising this automaton. In this respect, the games presented here improve over the focus games by showing
how small deterministic B\"uchi automata suffice for this task.

Table~\ref{fig:cmp ctl} tabulates the comparison of the \ctl satisfiability games with these other approaches.

\section{Further Work}
\label{sec:further}

The results of the previous section show that the game/automata approach to deciding \ctlstar is
reasonably viable in practice. Note that the implementation so far only features optimisations on one
of three fronts: it uses the latest and optimised technology for solving the resulting games. However,
there are two more fronts for optimisations which have not been exploited so far. The main advantage of
this approach is---as we believe---the combination of tableau-, automata- and game-machinery and
therefore the possible benefit from optimisation techniques in any of these areas. It remains to be
seen for instance whether the automaton determinisation procedure can be improved or replaced by a
better one. Also, the tableau community has been extremely successful in speeding up tableau-based
procedures using various optimisations. It also remains to be seen how those can be incorporated in the
combined method.

Furthermore, it remains to expand this work to extensions of \ctlstar, for example
\ctlstar with past operators, multi-agent logics based on \ctlstar, etc.

%
%
%

\bibliographystyle{alpha}
\bibliography{literature}

\end{document}